\documentclass[11pt,english]{article}

\usepackage[T1]{fontenc}
\usepackage[latin9]{inputenc}
\usepackage{geometry}
\usepackage{color}
\usepackage{array}
\usepackage{verbatim}
\usepackage{float}
\usepackage{booktabs}
\usepackage{amsmath}
\usepackage{amsthm}
\usepackage{amssymb}
\usepackage{graphicx}
\usepackage{rotfloat}
\usepackage{setspace}
\usepackage[authoryear]{natbib}
\makeatletter

\providecommand{\tabularnewline}{\\}

\theoremstyle{plain}
\newtheorem{thm}{\protect\theoremname}
\theoremstyle{plain}
\newtheorem{lem}{\protect\lemmaname}

\setcitestyle{round}
\usepackage{lscape}
\usepackage{longtable}

\usepackage[colorlinks,
            linkcolor=blue,
            anchorcolor=blue,
            citecolor=blue]{hyperref}

\usepackage{babel}
\providecommand{\lemmaname}{Lemma}
\providecommand{\theoremname}{Theorem}

\makeatother

\usepackage{babel}
\providecommand{\lemmaname}{Lemma}
\providecommand{\theoremname}{Theorem}

\begin{document}
\title{Realized GARCH, CBOE VIX, and the Volatility Risk Premium\thanks{Zhuo Huang acknowledges financial support by the National Natural
Science Foundation of China (71671004) and Tianyi Wang acknowledges
financial support by the National Natural Science Foundation of China
(71871060).}}
\author{\textbf{Peter Reinhard Hansen$^{a}$ \& Zhuo Huang$^{b}$ \& Chen
Tong}$^{c}$\textbf{ \& Tianyi Wang$^{d}$ }\bigskip{}
 \\
 \emph{\normalsize{}$^{a}$University of North Carolina \& Copenhagen
Business School}{\normalsize{}}\\
\emph{\normalsize{}$^{b}$Peking University, National School of Development}{\normalsize{}}\\
{\normalsize{}$^{c}$}\emph{\normalsize{}Department of Finance, School
of Economics \& Wang Yanan Institute} \emph{\normalsize{}for}\\
\emph{\normalsize{} Studies in Economics (WISE), Xiamen University}\\
 {\normalsize{}$^{d}$}\emph{\normalsize{}University of International
Business and Economics School of Banking and Finance\medskip{}
 }}
\date{\emph{\normalsize{}\today}}

\maketitle

\begin{abstract}
We show that the Realized GARCH model yields close-form expression
for both the Volatility Index (VIX) and the volatility risk premium
(VRP). The Realized GARCH model is driven by two shocks, a return
shock and a volatility shock, and these are natural state variables
in the stochastic discount factor (SDF). The volatility shock endows
the exponentially affine SDF with a compensation for volatility risk.
This leads to dissimilar dynamic properties under the physical and
risk-neutral measures that can explain time-variation in the VRP.
In an empirical application with the S\&P 500 returns, the VIX, and
the VRP, we find that the Realized GARCH model significantly outperforms
conventional GARCH models.

\bigskip{}
\end{abstract}
\textit{\small{}Keywords:}{\small{} Realized GARCH; High Frequency
Data; Volatility Risk Premium; Realized Variance; VIX.}{\small\par}

\noindent \textit{\small{}JEL Classification:}{\small{} C10; C22;
C80}{\small\par}

\noindent {\small{}\newpage}{\small\par}

\section{Introduction}

The variance risk premium or volatility risk premium (VRP) has been
the focus of much research since the seminal papers by \citet{CovalShumway2001}
and \citet{BK2003}. The VRP is the difference between the expected
return variance under the risk neutral measure and the expected return
variance under the physical measure,\footnote{This definition of the VRP follows that in \citet{BollerslevTauchenZhou2009}.
Other definitions are used in part of the literature.} where the former can be inferred from option prices. The leading
example is the VIX, which is a model-free measure of the expected
variance over the next 30 days under the risk neutral measure. Expectations
under the physical measure can be based on a suitable volatility model.\footnote{For instance a GARCH model or a reduced-form model for the realized
volatility. For alternative methods for computing the expected variance
under both the physical and risk neutral measures, see \citet{BollerslevTauchenZhou2009},
\citet{BollerMZ2011}, \citet{Wu_JOE2011} and \citet{Conrad_Loch2015}.} The VRP is a measure of volatility risk compensation and it is typically
positive. This is to be expected because large increases in volatility
tend to coincide with large negative returns. This relationship is
observed for a broad range of financial assets, see e.g. \citet{CovalShumway2001},
\citet{BK2003} and \citet{CW2009}. The VRP is also recognized as
a distinct risk factor that predicts both aggregate stock returns
and the cross-section of stock returns, see e.g. \citet{BollerslevTauchenZhou2009},
\citet{Bekaert2014}, and \citet{Cremers_HW2015}.%
{} Moreover, the VRP plays an important role in option pricing, see
e.g. \citet{Byun_JMY2015} and \citet{Song_Xiu2016}. Given its importance,
it is interesting to develop an econometric model that can explain
the time variation in the VRP while being coherent with other features
of the data.

Conventional GARCH models specify expectations under the physical
measure, $\mathbb{P}$, and additional structure is needed before
expectations can be computed under the risk neutral measure, $\mathbb{Q}$.
\citet{Duan1995} pioneered the use of GARCH models for option pricing
by introducing a locally risk-neutral valuation relationship (LRNVR).
The LRNVR defines a link between expected volatility under the $\mathbb{P}$
and expected volatility under $\mathbb{Q}$. With this additional
structure in place, GARCH models can be used to price options and
the corresponding VRP can be inferred. Unfortunately, standard GARCH
models combined with LRNVR cannot adequately explain the VRP, as shown
by \citet{HaoZhang2013}. They found that the VIX implied by GARCH
models is substantially below the observed VIX. \citet{HaoZhang2013}
explored if this shortcoming could be amended by modifying the objective
function to also target the VIX. Unfortunately, this leads to parameter
values (in the GARCH model) that contradict the empirical properties
under $\mathbb{P}$. In our empirical application, we also reach the
conclusion that GARCH models in conjunction with LRNVR cannot explain
the dynamic properties under both probability measures. A strong argument
for looking beyond standard GARCH models is provided by their diffusion
limits. These reveal that GARCH models are unable to fully compensate
for volatility risk, because GARCH models lack a separate volatility
shock variable. Stochastic volatility (SV) models, such as those by
\citet{Taylor1986} and \citet{Kim1998}, are better suited for this
situation because they have a dual-shock structure with distinct shocks
to returns and volatilities. This property facilitates a distinct
compensation for volatility risk, see e.g. \citet{BollerMZ2011}.
The main drawback of SV models is that they are more involved to estimate
than observation-driven models, such as GARCH models.\footnote{The same complication arises with Jump-GARCH models that typically
rely on particle filters for estimation, see e.g. \citet{Ornthanalai2014}. }

It is evident that the GARCH framework must be generalized in order
to become a coherent model of $\mathbb{P}$ and $\mathbb{Q}$. This
requires either a more flexible volatility model or a more sophisticated
risk neutralization method. In this paper, we pursue both extensions
by combining the Realized GARCH model with an exponentially affine
stochastic discount factor. This framework includes compensation for
both equity risk and volatility risk and it yields closed-form expressions
for both the VIX and the VRP. The Realized GARCH model is an observation-driven
model that conveniently has a dual shock structure that is similar
to that of SV models. This model is simple to estimate and easy to
combine with an\textcolor{blue}{{} }exponentially affine stochastic
discount factor. The parameter estimation can be adapted to include
VIX pricing errors in the objective function, similar to the estimation
method used in \citet{bardg2019}, see also \citet{AFT2019} who links
the realized volatility to volatility in an SV model. The estimated
model has several interesting properties. It delivers a higher level
of volatility, a higher volatility-of-volatility, and a stronger (more
negative) leverage correlation under $\mathbb{Q}$ than under $\mathbb{P}$.
The estimated model also generates higher levels of skewness and kurtosis
in cumulative returns under $\mathbb{Q}$ than under $\mathbb{P}$,
which are key determinants of the VRP, see \citet{Bakshi-Madan2006}
and \citet{Chabi-Yo2012}. The difference between log-volatility under
$\mathbb{P}$ and $\mathbb{Q}$ can conveniently be decomposed into
two terms, where the first term is compensation for equity risk through
the leverage effect and the second term is compensation for volatility-of-volatility. 

In an empirical analysis of daily S\&P 500 returns, realized volatilities,
and the VIX over 15 years, we compare the proposed model with a range
of alternative specifications. These include the EGARCH model by \citet{Nelson91},
the GARCH model by \citet{bollerslev:86}, Heston-Nandi GARCH by \citet{HN2000}.
These models are combined with either the the LRNVR by \citet{Duan1995}
or the variance dependent SDF by \citet{VDPK2013}.\footnote{\citet{VDPK2013}, introduced a variance-dependent SDF to improve
the option pricing performance of the Heston-Nandi GARCH model. The
idea was also used in \citet{Byun_JMY2015} in the context of Jump-GARCH
models.} We find that the new model has the best in-sample and out-of-sample
VIX pricing performance, and the proposed model does particularly
well during the turmoil period with the global financial crisis. We
find that the Realized GARCH model provides the best empirical fit
and, importantly, provides superior out-of-sample forecast of all
variables of interest.

The improved empirical results are driven by the inclusion of realized
volatility in the modeling. The help in two ways. First, the inclusion
improves volatility forecasts and this greatly improve the log-likelihood
of returns under $\mathbb{P}$. Second, the inclusion of a realized
volatility enables us to define a volatility shock that serves as
a second state variables in the SDF. This state variable characterizes
the compensation for volatility risk, which is important for explaining
key differences between $\mathbb{P}$ and $\mathbb{Q}$ and time-variation
therein.\footnote{The need for a model to simultaneously explain the variation under
$\mathbb{P}$ and $\mathbb{Q}$ was pointed out in \citet{Bates1996},
and has since received much attention in the option pricing literature,
see, e.g., \citet{PAN20023}, \citet{Eraker2004}, and \citet{SantaYan2010}. }

The remainder of this paper is organized as follows. We present the
Realized GARCH model, the risk-neutralization, and the model implied
VIX/VRP formula in Section 2 and discuss the distinct model dynamics
under $\mathbb{P}$ and $\mathbb{Q}$ in Section 3. In Section 4,
we present the set of competing models and we present our empirical
analysis in Section 5. We conclude in Section 6, present all proofs
in Appendix A, and present a range of empirical robustness checks
in Appendix B, which are based on different definitions of the VRP,
different error specification, and different choice of realized volatility
measure in the modeling.

\section{The Realized GARCH Model and VIX Pricing}

The Realized GARCH framework is a join model of returns and realized
volatility measures. Returns are modeled with a GARCH model, which
is augmented to include a realized measure of volatility, and the
Realized GARCH framework is characterized by a measurement equation
that ties the realized measure to the conditional variance. Realized
measures of volatility are computed from high frequency data where
the realized variance (RV) and the realized kernel (RK) by \citet{BHLS2008}
are prime examples. Realized GARCH models are generally found to outperform
conventional GARCH models in terms of modeling returns as well as
forecasting volatility. The reason is simply that the realized measures
provide more accurate measurements of volatility than daily returns,
and conventional GARCH models rely on the latter for ``updating''
the time-varying volatility. The Realized GARCH framework was introduced
by \citet{HansenHuangShek:2012} and later refined in \citet{RealizedEGARCH2016}
to have a more flexible leverage function and to allow for the inclusion
of multiple realized measures.

\subsection{Model under the Physical Measure}

We adapt the model in \citet{RealizedEGARCH2016} to the present context,
by adding an appropriate compensation for equity risk. Under the physical
measure the model is characterized by:
\begin{eqnarray}
r_{t} & = & r+\lambda\sqrt{h_{t}}-\tfrac{1}{2}h_{t}+\sqrt{h_{t}}z_{t},\label{eq:return.eq}\\
\log h_{t+1} & = & \omega+\beta\log h_{t}+\tau(z_{t})+\gamma\sigma u_{t},\label{eq:garch.eq}\\
\log x_{t} & = & \kappa+\phi\log h_{t}+\delta(z_{t})+\sigma u_{t},\label{eq:measurement.eq}
\end{eqnarray}
where $r_{t}$ is the logarithmic return, $\lambda$ is the price
of equity risk, $h_{t}=\mathrm{var}_{t-1}(r_{t})$ is the conditional
variance, $r$ is the risk-free interest rate, $z_{t}=(r_{t}-\mathbb{E}_{t-1}r_{t})/\sqrt{h_{t}}$
is the standardized return, and $x_{t}$ is the realized measure of
volatility. Our addition to this model framework, is a compensation
for equity risk that adds the term, $\lambda\sqrt{h_{t}}-\tfrac{1}{2}h_{t}$,
to the return equation (\ref{eq:return.eq}). The two random innovations,
$z_{t}$ and $u_{t}$, that are assumed to be independent and iid
standard Gaussian, $N(0,1)$. The quadratic functions $\tau(z)=\tau z+\tau_{2}(z^{2}-1)$
and $\delta(z)=\delta_{1}z+\delta_{2}(z^{2}-1)$ are \emph{leverage
functions} that capture dependence between return shocks and volatility
shocks. The parameter $\sigma$ can be interpreted as the volatility-of-volatility
shock. The model simplifies to a variant of the classical EGARCH\textcolor{black}{{}
model of \citet{Nelson91}} when $\gamma=0$.\footnote{We use $z_{t}^{2}$ in place of $|z_{t}|$ that was used the original
EGARCH model, which has some advantages, see \citet{HansenHuangShek:2012}
and \citet{RealizedEGARCH2016}. For completeness, we have also estimated
a Realized GARCH model with $\tau_{1}z_{t}+\tau_{2}(|z_{t}|-\sqrt{2/\pi})$,
which led to very similar qualitative and quantitative results.}

A key property of the model is that two shocks, $z_{t}$ and $u_{t}$,
are included in the GARCH equation, (\ref{eq:garch.eq}). This is
contrast to conventional GARCH models, where the conditional volatility
is solely driven by lagged daily returns.\footnote{For derivative pricing, this ``single shock'' structure is a serious
limitation because the equity risk premium parameter, $\lambda$,
must be increased to unreasonable high levels in order to explain
the variance risk premium, see \citet{HaoZhang2013} and our empirical
results in Table 2.} The dual-shock structure is important for describing the dynamic
properties under both $\mathbb{P}$ and $\mathbb{Q}$ simultaneously,
such that the dynamic properties of returns, the VIX, and the VRP,
can be explained within a unified coherent framework. 

\subsection{Risk Neutralization and Properties under the Risk Neutral Measure}

Before we can price the VIX, we need to state how the physical measure,
$\mathbb{P}$, relates to the risk neutral counterpart, $\mathbb{Q}$.
In the literature on option pricing with GARCH models, the most commonly
used risk neutralization methods are the locally risk-neutral valuation
relationship (LRNVR) by \citet{Duan1995} and the variance-dependent
SDF by \citet{VDPK2013}. These methods are applicable to a single-shock
GARCH models and do not apply to the dual-shock structure in our framework.
We will instead adopt an exponentially affine stochastic discount
factor, which has previously been used for risk neutralization with
multiple shocks in \citet{Corsi2013JFE}.

The stochastic discount factor, $M_{t+1}$, must satisfy $\mathbb{E}_{t}^{\mathbb{Q}}[X_{t+1}]=\mathbb{E}_{t}^{\mathbb{P}}[M_{t+1}X_{t+1}]$
for all asset prices, $X_{t+1}$. In the Realized GARCH framework
it is natural to use $z_{t+1}$ and $u_{t+1}$ as state variables,
and we will adopt the SDF define by:

\begin{equation}
M_{t+1}=\frac{\exp(-\lambda z_{t+1}-\xi u_{t+1})}{\mathbb{E}_{t}^{\mathbb{P}}\exp(-\lambda z_{t+1}-\xi u_{t+1})}=\exp\left\{ -\lambda z_{t+1}-\xi u_{t+1}-\tfrac{1}{2}(\lambda^{2}+\xi^{2})\right\} .\label{eq:SDF}
\end{equation}
Empirically, one would expect $\lambda$ to be positive and $\xi$
to be negative, which correspond to a positive equity premium and
a negative variance risk premium, respectively. It should be noted
that the parameter, $\lambda$, that appears in (\ref{eq:SDF}) is
identical to the $\lambda$ in the return equation, (\ref{eq:return.eq}).
This is not by assumption but an implication of a no arbitrage condition.
If we, as a starting point, permitted the $\lambda$ in (\ref{eq:SDF}),
to be a free and, possibly, time-varying parameter, then it can be
shown that this coefficient must be constant and equal to $\lambda$
in (\ref{eq:return.eq}). This is a consequence of a no-arbitrage
condition, see Lemma \ref{lem:LemConstantLambda} in the Appendix.

\subsubsection{Dynamic Properties under the Risk Neutral Measure}

Under the risk neutral measure, $\mathbb{Q}$, it follows that the
moment generating function (MGF) is given by 
\begin{eqnarray*}
\Psi(s_{1},s_{2})=\mathbb{E}_{t}^{\mathbb{Q}}[\exp(s_{1}z_{t+1}+s_{2}u{}_{t+1})] & = & \mathbb{E}_{t}^{\mathbb{P}}[M_{t+1}\exp(s_{1}z_{t+1}+s_{2}u_{t+1})]\\
 & = & \exp[-s_{1}\lambda-s_{2}\xi+\tfrac{1}{2}(s_{1}^{2}+s_{2}^{2})].
\end{eqnarray*}
This MGF is identical to $\mathbb{E}_{t}^{\mathbb{P}}[\exp(s_{1}z_{t+1}^{\ast}+s_{2}u_{t+1}^{\ast})]$,
where $z_{t+1}^{\ast}=z_{t+1}+\lambda$ and $u_{t+1}^{*}=u_{t+1}+\xi$,
and it implies the following dynamic model under the risk neutral
measure:{\small{}
\begin{eqnarray}
r_{t+1} & = & r-\tfrac{1}{2}h_{t+1}+\sqrt{h_{t+1}}z_{t+1}^{\ast},\label{eq:returnQ}\\
\log h_{t+1} & = & \tilde{\omega}+\beta\log h_{t}+\tilde{\tau}_{1}z_{t}^{\ast}+\tau_{2}(z_{t}^{\ast2}-1)+\gamma\sigma u_{t}^{\ast},\label{eq:GARCHq}\\
\log x_{t} & = & \tilde{\kappa}+\phi\log h_{t}+\tilde{\delta}_{1}z_{t}^{\ast}+\delta_{2}(z_{t}^{\ast2}-1)+\sigma u_{t}^{\ast},\label{eq:measureQ}
\end{eqnarray}
}where $(z_{t}^{\ast},u_{t}^{\ast})$ has bivariate Gaussian distribution,
$N(0,I)$, under $\mathbb{Q}$. The relationships between parameters
(under $\mathbb{P}$ and $\mathbb{Q}$) are: $\tilde{\omega}=\omega-\tau_{1}\lambda+\tau_{2}\lambda^{2}-\gamma\sigma\xi$,
$\tilde{\tau}_{1}=\tau_{1}-2\tau_{2}\lambda$, $\tilde{\kappa}=\kappa-\delta_{1}\lambda+\delta_{2}\lambda^{2}-\sigma\xi$,
and $\tilde{\delta}_{1}=\delta_{1}-2\delta_{2}\lambda$, see Lemma
\ref{LemRGunderQ} for details.

The mapping $(z_{t},u_{t})\mapsto(z_{t}^{\ast},u_{t}^{\ast})$ can
be viewed as a generalization of LRNVR to the bivariate case. In the
present context, this bivariate structure rely on the inclusion of
the realized measure in the modeling, which facilitates a more complex
dynamic structure than can be achieved with conventional GARCH models.
The conventional GARCH model emerges as a special case when $\gamma=0$
and $\xi=0$. In this situation, (\ref{eq:returnQ}) and (\ref{eq:GARCHq})
simplify to an EGARCH model with the change of measure, $z_{t+1}^{*}=z_{t+1}+\lambda$.
This reveals a close relation between the GARCH models with exponentially
affine SDF and the simple change of measure that was proposed by \citet{Duan1995}.
This connection appears to have been overlooked in the exiting literature.

\subsection{The VIX Pricing Formula}

The Chicago Board Options Exchange's (CBOE) VIX index is defined as
the square-root of the annualized expected variance over the next
30 calendar days, where the expectation is computed under the risk
neutral measure. We will model returns and realized measure using
daily data and we will use 22 trading days and 252 trading days to
represent a month and a year, respectively. The model-based VIX formula
is therefore given by
\[
\mathrm{VIX}_{t}^{\mathrm{model}}=\sqrt{\tfrac{252}{22}\sum_{k=1}^{22}\mathbb{E}_{t}^{\mathbb{Q}}(h_{t+k})}\times100,
\]
where the expression for $\mathbb{E}_{t}^{\mathbb{Q}}(h_{t+k})$ is
model specific. The combination of the Realized GARCH model and the
exponentially affine SDF leads to the following expression:
\begin{thm}
\label{theo:VIXpricingRG}For the Realized GARCH model (\ref{eq:return.eq})-(\ref{eq:measurement.eq})
and the SDF (\ref{eq:SDF}), the model-implied VIX is given by:{\small{}
\begin{equation}
\mathrm{VIX}_{t}^{\mathrm{RG}}=100\times\sqrt{\frac{252}{22}\left[h_{t+1}+\sum_{k=2}^{22}\left(\prod_{i=0}^{k-2}F_{i}\right)h_{t+1}^{\beta^{k-1}}\right]},\label{eq:VIXpricingRG}
\end{equation}
where $F_{i}=(1-2\beta^{i}\tau_{2})^{-1/2}\exp\left\{ \beta^{i}(\tilde{\omega}-\tau_{2})+\tfrac{1}{2}\beta^{2i}[\tfrac{\tilde{\tau}_{1}^{2}}{1-2\beta^{i}\tau_{2}}+\gamma^{2}\sigma^{2}]\right\} $,
with $\tilde{\omega}=\omega-\tau_{1}\lambda+\tau_{2}\lambda^{2}-\delta\sigma\xi$
and $\tilde{\tau}_{1}=\tau_{1}-2\lambda\tau_{2}$.}{\small\par}
\end{thm}
The expression (\ref{eq:VIXpricingRG}) facilitates an easy comparison
of the model implied VIX with the actual VIX index, and it is analogous
to the expressions obtained for a range of conventional GARCH models
in \citet{HaoZhang2013}. The proof of Theorem \ref{theo:VIXpricingRG}
is given in Appendix \ref{sec:Appendix-of-Proofs}. 

\subsection{Volatility risk premium}

The literature has proposed several definitions of the VRP, see \citet{BollerslevTauchenZhou2009}.\footnote{See \citet{BollerslevTauchenZhou2009} for detailed discussion of
VRP measures, including ex-post and ex-ante measures, and measures
in units of variances and in units of volatilities.} We adopt the following definition,

\[
\mathrm{VRP}_{t}^{\mathtt{modelfree}}=\mathrm{VIX}_{t}-\sqrt{\tfrac{252}{22}\mathbb{E}_{t}^{\mathbb{P}}\left(\sum_{i=1}^{22}\mathrm{RVcc}_{t+i}\right)}\times100,
\]
where $\mathrm{RVcc}_{t}=\mathrm{RV}_{t}+r_{co,t}^{2}$, $\mathrm{RV}_{t}$
is the realized variance estimator for the hours with active trading
on day $t$, and $r_{co,t}^{2}$ is the squared overnight return,
which is computed from the closing price on day $t-1$ and the opening
price of day $t$. This difference between the observed VIX and the
expected realized measure of volatility (for the corresponding 22
trading days) is a model-free measure of the VRP. This is a theoretical
quantity, because the expectation operator depends on $\mathbb{P}$
that is unknown in practice. The following empirical VRP 

\[
\mathrm{VRP}_{t}^{\mathtt{market}}=\mathrm{VIX}_{t}-\sqrt{\tfrac{252}{22}\sum_{i=1}^{22}\mathrm{RVcc}_{t-i+1}}\times100,
\]
was proposed in \citet{BollerslevTauchenZhou2009}. This quantity
relies on the assumption that realized volatility follows a martingale
process, such that the expected monthly realized volatility is given
by the observed realized volatility over the most recent month. As
a robustness check, we also consider a second, alternative, empirical
measure, which is based on the heterogeneous autoregressive (HAR)
model by \citet{Corsi2009}, see Appendix \ref{Sec:RobustalterVRP}. 

Our model-implied VRP is simply 
\[
\mathrm{VRP}_{t}^{\mathtt{model}}=\left(\sqrt{\tfrac{252}{22}\sum_{k=1}^{22}\mathbb{E}_{t}^{\mathbb{Q}}(h_{t+k})}-\sqrt{\tfrac{252}{22}\sum_{k=1}^{22}\mathbb{E}_{t}^{\mathbb{P}}(h_{t+k})}\right)\times100,
\]
which is the annualized, one-month ahead, expected volatility using
the risk neutral measure less the corresponding quantity under the
physical probability measure.

\section{Key Model Properties under $\mathbb{P}$ and $\mathbb{Q}$}

In this section, we analyze the Realized GARCH model with the exponentially
affine SDF, under both $\mathbb{P}$ and $\mathbb{Q}$, and we derive
key properties of volatility, leverage, and returns under both measures.
These results provide theoretical insight about importance of various
model parameters and their interpretations. We derive the properties
under the assumption that the parameters in the Realized GARCH model
and the SDF satisfy: 
\[
|\beta|<1,\quad\lambda,\gamma,\sigma,\tau_{2},\delta_{2}>0,\quad\xi,\tau_{1},\delta_{1}<0.
\]
These inequalities guarantee the following properties: (a) the volatility
process is stationary, (b) the equity premium is positive, (c) the
volatility premium is negative, and (d) the model has a leverage effect.
The stated parameter restrictions are consistent with our empirical
results in Section \ref{sec:Empirical-Analysis}.

\subsection{Average Volatility}

For the average log-volatility (unconditional mean of $\log h_{t}$)
we have that
\[
\mathbb{E}^{\mathbb{P}}(\log h)=\frac{\omega}{1-\beta}\qquad\mbox{and}\qquad\mathbb{E}^{\mathbb{Q}}(\log h)=\frac{\omega-\tau_{1}\lambda+\tau_{2}\lambda^{2}-\gamma\sigma\xi}{1-\beta}.
\]
Given our assumptions above, it follows that the average log-volatility
is higher under the $\mathbb{Q}$-measure than under the $\mathbb{P}$-measure.
Thus, the logarithmic variant of the VRP, 

\begin{equation}
\mathbb{E}^{\mathbb{Q}}(\log h)\text{-}\mathbb{E}^{\mathbb{P}}(\log h)=\frac{-\tau_{1}\lambda+\tau_{2}\lambda^{2}-\gamma\sigma\xi}{1-\beta}=\frac{-\tau_{1}\lambda+\tau_{2}\lambda^{2}}{1-\beta}+\frac{-\gamma\sigma\xi}{1-\beta},\label{eq:vrp_decom}
\end{equation}
is positive. The logarithmic variant of the VRP was used in \citet{CW2009}
and \citet{AMMANN2013}.

From (\ref{eq:vrp_decom}) we see that the logarithmic VRP can be
decomposed into two terms. One that is driven by the equity risk premium
and the leverage effect and a second term that is driven by the compensation
for volatility risk and volatility-of-volatility due to volatility
shocks. Their relative contributions to the log VRP are given by $\frac{-\tau_{1}\lambda+\tau_{2}\lambda^{2}}{-\tau_{1}\lambda+\tau_{2}\lambda^{2}-\gamma\sigma\xi}$
and $\frac{-\gamma\sigma\xi}{-\tau_{1}\lambda+\tau_{2}\lambda^{2}-\gamma\sigma\xi}$
, respectively, which, conveniently, do not depend on $\beta$. 

It is worth noting that a positive VRP does not require the equity
risk premium to be positive if $\xi$ is sufficiently negative. The
literature typically finds $\lambda>0$, e.g. \citet{French1987},
but negative values have also been reported, see e.g. \citet{JLA2001}.

\subsection{Volatility-of-Volatility}

The volatility-of-volatility (in log-$h$) can be derived similarly
and is, under $\mathbb{P}$ and $\mathbb{Q}$, given by
\[
\mathrm{var}^{\mathbb{P}}(\log h)=\frac{\tau_{1}^{2}+2\tau_{2}^{2}+\gamma^{2}\sigma^{2}}{1-\beta}\qquad\mbox{and}\qquad\mathrm{var}^{\mathbb{Q}}(\log h)=\frac{(\tau_{1}-2\tau_{2}\lambda)^{2}+2\tau_{2}^{2}+\gamma^{2}\sigma^{2}}{1-\beta^{2}},
\]
respectively. Since $\tau_{1}<0$ it follows that their difference,
$\mathrm{var}^{\mathbb{Q}}(\log h)-\mathrm{var}^{\mathbb{P}}(\log h)=4\lambda(\lambda\tau_{2}^{2}-\tau_{1}\tau_{2})/(1-\beta^{2})$
is positive, such that the unconditional variance of volatility is
higher under $\mathbb{Q}$ than under $\mathbb{P}$.

\subsection{Dependence between Returns and Volatility (Leverage)}

The dependence between returns and volatility is another important
aspect of asset pricing. Here we follow \citet{Christoffersen2014}
and quantify this dependence using the conditional correlation between
$\log h_{t+1}$ and $r_{t}$, (leverage correlations) under $\mathbb{P}$
and $\mathbb{Q}$. Under $\mathbb{P}$ we have
\[
\mathrm{cov}_{t}^{\mathbb{P}}\left(\log h_{t+1},r_{t}\right)=\mathbb{E}_{t}^{\mathbb{P}}[(\tau_{1}z_{t}+\tau_{2}z_{t}^{2}+\gamma\sigma u_{t})z_{t}\sqrt{h_{t}}]=\tau_{1}\sqrt{h_{t}},
\]
such that the conditional correlation is 
\[
\rho_{\mathbb{P}}=\mathrm{corr}_{t}^{\mathbb{P}}\left(\log h_{t+1},r_{t}\right)=\frac{\tau_{1}}{\sqrt{\tau_{1}^{2}+2\tau_{2}^{2}+\gamma^{2}\sigma^{2}}}.
\]
Under the risk neutral measure, $\mathbb{Q}$, we find that $\mathrm{cov}_{t}^{\mathbb{Q}}\left(\log h_{t+1},r_{t}\right)=(\tau_{1}-2\tau_{2}\lambda)\sqrt{h_{t}}$
such that
\[
\rho_{\mathbb{Q}}=\mathrm{corr}_{t}^{\mathbb{Q}}\left(\log h_{t+1},r_{t}\right)=\frac{\tau_{1}-2\lambda\tau_{2}}{\sqrt{(\tau_{1}-2\tau_{2}\lambda)^{2}+2\tau_{2}^{2}+\gamma^{2}\sigma^{2}}}.
\]
These correlations are, as expected, both negative, and it can be
shown that $\rho_{\mathbb{Q}}^{2}-\rho_{\mathbb{P}}^{2}>0$, such
that the leverage effect is more pronounce under $\mathbb{Q}$ than
under $\mathbb{P}$.

\subsection{Skewness and Kurtosis of Multi-period Returns}

While VIX pricing only requires the expectations of future volatility,
many other problems, such as option pricing, require an accurate description
of the distribution of cumulative returns.\footnote{For example, \citet{DGS1999} provided a method to price options with
the skewness and kurtosis of cumulative returns.} Figure \ref{fig:SandK} presents the skewness and kurtosis of cumulative
returns for the Realized GARCH model, for cumulative returns spanning
a period from 1 to 250 days (approximately one year). For comparison,
we also include the corresponding results based on the EGARCH model.
The simulation designs for the two models are the estimates we obtained
in our empirical analysis, see Table 2. Because closed-form expressions
for skewness and kurtosis of cumulative returns are not readily available,
these results are based on simulation methods with 1,000,000 replications.
The first 750 days were discarded in each simulation in order to minimize
the influence of initial values.

The results in Figure \ref{fig:SandK} show that cumulative returns
are more left-skewed (have a more negative skewness) under $\mathbb{Q}$
than under $\mathbb{P}$, and the tails are also thicker (larger kurtosis)
under $\mathbb{Q}$ than under $\mathbb{P}$. This is true for both
the Realized GARCH model and the EGARCH model. However, the magnitude
of skewness and kurtosis is much larger for the Realized GARCH model
especially under the risk neutral measures. These features of the
Realized GARCH model are potentially important because theoretical
results in \citet{Bakshi-Madan2006} and \citet{Chabi-Yo2012} demonstrate
that the skewness and the kurtosis of the market index are key determinants
of the variance risk premium.
\begin{figure}[H]
\begin{centering}
\includegraphics[scale=0.5]{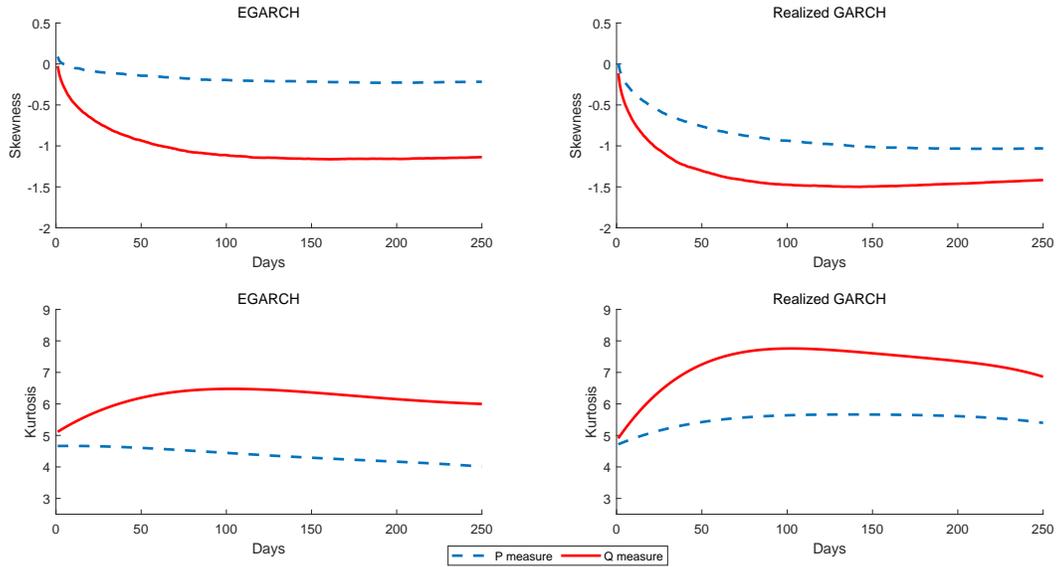}
\par\end{centering}
\caption{Skewness and kurtosis of cumulative returns under $\mathbb{P}$ and
$\mathbb{Q}$ for the EGARCH and the Realized GARCH model.\label{fig:SandK}}
\end{figure}

In Figure \ref{fig:SimDensity}, we present the simulated densities
for standardized cumulative returns over one month (left panels) and
six months (right panels). The densities under $\mathbb{P}$ are in
the upper panels and those under $\mathbb{Q}$ are presented in the
lower panels, where the solid red lines are for the Realized GARCH
model and the dashed blue line are for the EGARCH model based on the
parameter estimates we obtained in our empirical analysis. A left
skew can be seen for both models and it is more pronounced at longer
horizons (six months), especially for the Realized GARCH model. The
skewness is also more pronounced under the risk neutral measure, $\mathbb{Q}$,
which is consistent with the results in Figure \ref{fig:SandK}.
\begin{figure}[H]
\begin{centering}
\includegraphics[scale=0.5]{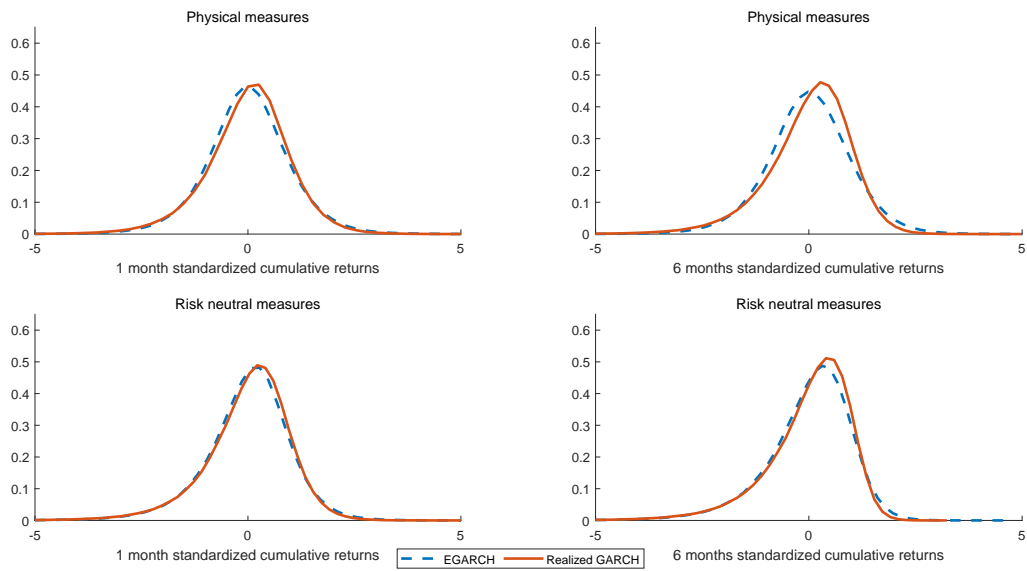}
\par\end{centering}
\caption{Density of Cumulative Returns via Simulation\label{fig:SimDensity}}
\end{figure}

\section{Some Competing Models and their Properties}

In this section, we introduce four alternative models, which we will
use to benchmark the Realized GARCH model against. We compare the
models ability to explain and predict the three variables: return
volatility, the market VIX, and the volatility risk premium. The four
alternative models are: the GARCH model by \citet{bollerslev:86},
the EGARCH model by \citet{Nelson91}, the Heston-Nandi GARCH model
by \citet{HN2000}, which are combined with Duan's LRNVR, and the
Heston-Nandi GARCH combined with the variance dependent SDF, as proposed
by \citet{VDPK2013}.\footnote{\citet{HaoZhang2013} examined GARCH, EGARCH, TGARCH, AGARCH and CGARCH
models. To conserve space, we focus on the GARCH and EGARCH models
because the EGARCH had the best performance in the study by \citet{HaoZhang2013},
and the original GARCH model is a natural benchmark.}

\subsection{GARCH and EGARCH Model}

\textcolor{black}{GARCH and EGARCH model are commonly used as benchmarks
in comparisons of volatility models. The original GARCH model tends
to perform well with exchange rate data, but it is typically outperformed
by models that can accommodate a leverage effect when applied to equity
returns, see \citet{HansenLundeBeatGarch}. The volatility dynamics
for the GARCH(1,1) is given by}
\[
h_{t+1}=\omega+\beta h_{t}+\alpha h_{t}z_{t}^{2},
\]
and that of the\textcolor{black}{{} EGARCH(1,1) is given by}
\begin{eqnarray*}
\log h_{t+1} & = & \omega+\beta\log h_{t}+\tau_{1}z_{t}+\tau_{2}(|z_{t}|-\sqrt{2/\pi}).
\end{eqnarray*}
From \citet[propositions 1-4]{HaoZhang2013} it follows that these
models in conjunction with the exponentially affine SDF yield the
following model-based pricing formulae for the VIX:{\small{}
\[
\mathrm{VIX}_{t}^{\mathrm{G}}=100\times\sqrt{252\sigma_{h}^{2}+\tfrac{252}{22}\frac{1-\beta^{22}}{1-\beta}(h_{t+1}-\sigma_{h}^{2})},\qquad\mbox{with}\quad\sigma_{h}^{2}=\omega/(1-\beta),
\]
}and

{\small{}
\[
\mathrm{VIX}_{t}^{\mathrm{EG}}=100\times\sqrt{\tfrac{252}{22}\left[h_{t+1}+\sum_{k=2}^{22}\left(\prod_{i=0}^{k-2}F_{i}\right)h_{t+1}^{\beta^{k-1}}\right]},
\]
}respectively, where{\small{} 
\begin{eqnarray*}
F_{i} & = & \exp\left[\beta\left(\omega-\tau_{2}\sqrt{\frac{2}{\pi}}\right)\right]\left\{ \exp\left[-\beta^{i}(\tau_{1}-\tau_{2})\lambda+\frac{\beta^{2i}(\tau_{2}-\tau_{1})^{2}}{2}\right]\Phi[\lambda-\beta^{i}(\tau_{1}-\tau_{2})]\right.\\
 &  & \left.+\exp\left[-\beta^{i}(\tau_{1}+\tau_{2})\lambda+\frac{\beta^{2i}(\tau_{2}+\tau_{1})^{2}}{2}\right]\Phi[\beta^{i}(\tau_{1}+\tau_{2})-\lambda]\right\} .
\end{eqnarray*}
}{\small\par}

\subsection{Heston-Nandi GARCH Model under LRNVR}

The Heston-Nandi GARCH model is a popular discrete-time model for
option pricing. The equity premium is assumed to be proportional to
the conditional variance and a specific leverage term is adopted in
the GARCH equation,
\begin{eqnarray*}
r_{t+1} & = & r+\lambda h_{t+1}-\tfrac{1}{2}h_{t+1}+\sqrt{h_{t+1}}z_{t+1},\\
h_{t+1} & = & \omega+\beta h_{t}+\alpha(z_{t}-\delta\sqrt{h_{t}})^{2}.
\end{eqnarray*}
This structure conveniently yields a closed-form option pricing formula.
Under LRNVR risk neutralization the corresponding dynamics under $\mathbb{Q}$
is:
\begin{eqnarray*}
r_{t+1} & = & r-\tfrac{1}{2}h_{t+1}+\sqrt{h_{t+1}}z_{t+1}^{\ast},\\
h_{t+1} & = & \omega+\beta h_{t}+\alpha(z_{t}^{\ast}-(\delta+\lambda)\sqrt{h_{t}})^{2},\\
 & = & \underset{=\tilde{\omega}}{\underbrace{\omega+\alpha}}+\underset{=\tilde{\beta}}{\underbrace{\left[\beta+\alpha(\delta+\lambda)^{2}\right]}}h_{t}-2\alpha(\delta+\lambda)\sqrt{h_{t}}z_{t}^{\ast}+\alpha(z_{t}^{\ast2}-1).
\end{eqnarray*}
where $z_{t}^{*}=z_{t}+\lambda\sqrt{h_{t}}$. If $|\tilde{\beta}|<1$,
then the unconditional mean of $h_{t}$ under $\mathbb{Q}$ is $\sigma_{h}^{2}=\tilde{\omega}/(1-\tilde{\beta})$,
and the $k$-step ahead expected conditional variance is

\[
\mathbb{E}_{t}^{\mathbb{Q}}(h_{t+k})=\sigma_{h}^{2}+\tilde{\beta}^{k-1}(h_{t+1}-\sigma_{h}^{2}).
\]
The model implied VIX pricing formula is therefore given by:

{\small{}
\[
\mathrm{VIX}_{t}^{\mathrm{HN}}=\sqrt{\tfrac{252}{22}\sum_{k=1}^{22}\mathbb{E}_{t}^{\mathbb{Q}}(h_{t+k})}\times100=\sqrt{252\sigma_{h}^{2}+\tfrac{252}{22}\tfrac{1-\tilde{\beta}^{22}}{1-\tilde{\beta}}(h_{t+1}-\sigma_{h}^{2})}\times100.
\]
}{\small\par}

\subsection{Heston-Nandi GARCH under Variance Dependent SDF}

An alternative to LRNVR is the variance dependent SDF by \citet{VDPK2013}.
As suggested by its name, this SDF depends on $h_{t}$, and this dependence
has been shown to improve the option pricing performance of the Heston-Nandi
GARCH model.

\citet{VDPK2013}, show that the dynamic properties under $\mathbb{Q}$
are given by
\begin{eqnarray*}
r_{t+1} & = & r-\tfrac{1}{2}h_{t+1}^{*}+\sqrt{h_{t+1}^{*}}z_{t+1}^{*},\\
h_{t+1}^{*} & = & \omega^{*}+\beta h_{t}^{*}+\alpha^{*}(z_{t}^{*}-\delta^{*}\sqrt{h_{t}^{*}})^{2},
\end{eqnarray*}
where $z_{t}^{*}$ has a standard normal distribution and 
\begin{eqnarray*}
h_{t}^{*} & = & h_{t}/(1+2\alpha\xi),\qquad\qquad\omega^{*}=\omega/(1+2\alpha\xi),\\
\alpha^{*} & = & \alpha/(1+2\alpha\xi)^{2},\qquad\quad\delta^{*}=(\lambda+\delta-\tfrac{1}{2})(1+2\alpha\xi)+\tfrac{1}{2}.
\end{eqnarray*}
Here $\xi$ is the variance risk aversion parameter, see \citet{VDPK2013}.
The resulting model-implied VIX pricing formula is given by

{\small{}
\[
\mathrm{VIX}_{t}^{\mathrm{HN}_{\mathtt{vd}}}=\sqrt{252\sigma_{h}^{*2}+\tfrac{252}{22}\tfrac{1-\beta^{\ast}{}^{22}}{1-\beta^{\ast}}(h_{t+1}^{*}-\sigma_{h}^{*2})},
\]
}where $\sigma_{h}^{*2}=(\omega^{\ast}+\alpha^{\ast})/(1-\beta^{\ast})$
with $\beta^{\ast}=\beta+\alpha^{\ast}\delta^{\ast2}$. Unlike LRNVR,
the variance dependent SDF will induce a transformation of one-step-ahead
conditional variance after the change of measure.\footnote{\citet{VDPK2013} suggested to reparametrize the model with $\tilde{\xi}=1/(1+2\alpha\xi)$
in place of $\xi$, when estimating the model.}

\section{Empirical Analysis\label{sec:Empirical-Analysis}}

For the estimation it is convenient to use a different expression
of the GARCH equation,
\begin{equation}
\log h_{t+1}=(\omega-\gamma\kappa)+(\beta-\gamma\phi)\log h_{t}+(\tau(z_{t})-\gamma\delta(z_{t}))+\gamma\log x_{t},\tag{{\ref{eq:garch.eq}'}}\label{eq:MSRG-measurement}
\end{equation}
 which is obtained by substituting (\ref{eq:measurement.eq}) into
(\ref{eq:garch.eq}). This formulation highlights the observation-driven
structure of the model, as it shows how the conditional volatility
depends on the observable realized measure and (a function of) the
lagged standardized return. This makes evaluation and maximization
of the log-likelihood straight forward.\footnote{For the related stochastic volatility models, direct maximization
of the likelihood for is typically impractical, and other estimation
methods, such as GMM and simulation based methods, are often employed
for this type of models.}

\subsection{Model Estimation}

We estimate the unknown parameters (of the model and the SDF) by maximizing
a joint log-likelihood function that is composed of the log-likelihood
function of the (Realized) GARCH model and the log-likelihood for
VIX pricing errors.

The log-likelihood function for the Realized GARCH model specifies
the dynamics for $(r_{t},x_{t})$ while the GARCH models (GARCH, EGARCH,
Heston-Nandi) specifies the dynamics for $r_{t}$. These likelihood
terms are combined with a log-likelihood for the VIX pricing errors,
where the latter is influenced by both the choice of volatility model
and the choice of SDF. Following \citet{HaoZhang2013} we adopt at
Gaussian specification for the pricing error, where $\mathrm{VIX}_{t}^{\mathrm{Model}}-\mathrm{VIX}_{t}\sim iidN(0,\sigma_{\mathtt{vix}}^{2})$,
and as a robustness check we also estimate the parameters using a
second (multiplicative) specification: $\mathrm{VIX}_{t}=\mathrm{VIX}_{t}^{\mathrm{Model}}\eta_{t}$,
where it is assumed that $\log\eta_{t}\sim iidN(-\sigma_{\mathtt{vix}}^{2}/2,\sigma_{\mathtt{vix}}^{2})$,
such that $\mathbb{E}(\eta_{t})=1$. The two specifications produce
very similar estimates and similar pricing errors, see Appendix \ref{Sec:RobustlogVIX}.

For the Realized GARCH model, the total (quasi) log-likelihood is
given by 
\[
\ell_{r}+\ell_{x}+\ell_{\mathtt{vix}},
\]
where 
\begin{eqnarray*}
\ell_{r} & = & -\tfrac{1}{2}\sum_{t=1}^{T}\{\log2\pi+\log h_{t}^{\mathtt{RG}}+(r_{t}-\mu_{t}^{\mathtt{RG}})^{2}/h_{t}^{\mathtt{RG}}\},\\
\ell_{x} & = & -\tfrac{1}{2}\sum_{t=1}^{T}\{\log2\pi+\log\sigma^{2}+[\log x_{t}-\omega-\beta\log h_{t}^{\mathtt{RG}}-\delta(z_{t})]^{2}/\sigma^{2}\},\\
\ell_{\mathtt{vix}} & = & -\tfrac{1}{2}\sum_{t=1}^{T}\{\log2\pi+\log\sigma_{\mathtt{vix}}^{2}+(\mathrm{VIX}_{t}^{\mathtt{RG}}-\mathrm{VIX}_{t})^{2}/\sigma_{\mathtt{vix}}^{2}\},
\end{eqnarray*}
with $h_{t}^{\mathtt{RG}}$ given from the GARCH equation (\ref{eq:garch.eq})
and $\mu_{t}^{\mathtt{RG}}=r+\lambda\sqrt{h_{t}^{\mathtt{RG}}}-\tfrac{1}{2}h_{t}^{\mathtt{RG}}$.

The likelihood of the other models are define similarly with model-specific
definitions of $\mu_{t}$, $h_{t}$, and the model-implied VIX. The
conventional GARCH models do not have the second term of the log-likelihood,
because they do include the realized measure, $x_{t}$, in the modeling.

The idea of combining the likelihood of a time-series model with a
second likelihood for option pricing errors is now standard in this
literature. Some papers including pricing errors for the a range of
options, see e.g. \citet{VDPK2013} and \citet{Christoffersen2014},
or pricing errors for volatility derivatives, see e.g. \citet{Wang2017},
\citet{bardg2019}, or VIX pricing errors as in \citet{HaoZhang2013}.
This is in contrast to an earlier literature that implicitly assumed
pricing errors to be zero and adopted the VIX as the volatility variables,
see, e.g., \citet{Duan2010JEDC} who estimated a stochastic volatility
model with jumps by exploiting the theoretical link between the VIX
and the latent volatility.

\subsection{Data}

Our empirical analysis is based on daily data for S\&P 500 stock index
and CBOE VIX. We obtain the daily VIX index and the daily returns
from Yahoo Finance while the realized measure are downloaded from
the Realized Library at Oxford-Man Institute. The primary realized
measure is the realized variance from the hours with active trading
with the squared overnight return added, see \citet{HansenLundeWholeDay}.
As another robustness check of our main results, we have also used
different choices of realized measures, see Appendix \ref{Sec:Robustalterrms}.

Our full sample spans 15 years, from January 2004 to December 2018.
We will present empirical results based on the full sample period
as well as out-of-sample results where the model is estimated recursively
using a rolling window sample with 750 days. The out-of-sample performance
is evaluated over the years 2007 to 2018. We also present separate
out-of-sample results for two subsamples: the years 2007-2012, which
include the global financial crisis period, and the years 2013-2018,
which span the post-crisis period.

\subsubsection{CBOE VIX calculation}

The VIX index is a model-free measure of volatility. Prior to being
annualized, it is computed as
\[
\mathrm{VIX}_{t}=\sqrt{\frac{2}{T}\sum_{i}\frac{\Delta K_{i}}{K_{i}^{2}}\exp\left(rT\right)Q\left(K_{i}\right)-\frac{1}{T}\left(\frac{F}{K_{0}}-1\right)^{2}},
\]
where $T$ is the time to maturity, $F$ is the forward index level,
$K_{0}$ is the first strike below $F$, $K_{i}$ is the strike price
of the $i$-th out-of-the-money option, $\Delta K_{i}$ is the interval
between strike prices, $r$ is the risk-free rate associated with
time to maturity, $Q(K_{i})$ is the midpoint of the bid-ask spread
for options with strike $K_{i}$. See \citet{JiangTian2005} for a
detail discussion on the VIX formula, and \citet{VIXreview2015} for
a review of model-free measures.\footnote{The VIX formula is described in the CBOE white paper, http://www.cboe.com/micro/vix/vixwhite.pdf,
and is based on earlier results in, \citet{CarrMadan1998}, \citet{DDKZ1999},
and \citet{BN2000}, who applied similar methods to approximate the
expected volatility under $\mathbb{Q}$.}
\begin{table}[H]
\caption{Summary of Statistics \label{Summary_of_statistics}}

\medskip{}

\begin{centering}
\textcolor{black}{}%
\begin{tabular}{r@{\extracolsep{0pt}.}lcccccc}
\toprule 
\multicolumn{2}{c}{} & \textcolor{black}{\small{}Mean} & \textcolor{black}{\small{}Median} & \textcolor{black}{\small{}Min} & \textcolor{black}{\small{}Q1} & \textcolor{black}{\small{}Q3} & \multicolumn{1}{c}{\textcolor{black}{\small{}Max}}\tabularnewline
\midrule 
\multicolumn{2}{c}{\textcolor{black}{\small{}$VIX$}} & \textcolor{black}{\small{}18.410} & \textcolor{black}{\small{}15.690} & \textcolor{black}{\small{}9.140} & \textcolor{black}{\small{}12.940} & \textcolor{black}{\small{}20.948} & \textcolor{black}{\small{}80.860}\tabularnewline
\multicolumn{2}{c}{\textcolor{black}{\small{}$\text{\ensuremath{\sqrt{AnnRV}}}$}} & \textcolor{black}{\small{}13.573} & \textcolor{black}{\small{}10.726} & \textcolor{black}{\small{}4.211} & \textcolor{black}{\small{}8.517} & \textcolor{black}{\small{}15.157} & \textcolor{black}{\small{}73.553}\tabularnewline
\multicolumn{2}{c}{\textcolor{black}{\small{}$Ret(\%)$}} & \textcolor{black}{\small{}0.017} & \textcolor{black}{\small{}0.058} & \textcolor{black}{\small{}-9.127} & \textcolor{black}{\small{}-0.378} & \textcolor{black}{\small{}0.485} & \textcolor{black}{\small{}10.246}\tabularnewline
\multicolumn{2}{c}{\textcolor{black}{\small{}$VRP$}} & \textcolor{black}{4.837} & \textcolor{black}{4.737} & \textcolor{black}{\small{}-25.284} & \textcolor{black}{\small{}3.147} & \textcolor{black}{\small{}6.510} & \textcolor{black}{\small{}28.316}\tabularnewline
\midrule 
\multicolumn{2}{c}{} & \textcolor{black}{\small{}Std} & \textcolor{black}{\small{}Skew} & \textcolor{black}{\small{}Kurt} & \textcolor{black}{\small{}AR1} & \textcolor{black}{\small{}AR10} & \textcolor{black}{\small{}AR22}\tabularnewline
\midrule
\multicolumn{2}{c}{\textcolor{black}{\small{}$VIX$}} & \textcolor{black}{\small{}8.812} & \textcolor{black}{\small{}2.657} & \textcolor{black}{\small{}12.662} & \textcolor{black}{\small{}0.980} & \textcolor{black}{\small{}0.898} & \textcolor{black}{\small{}0.810}\tabularnewline
\multicolumn{2}{c}{\textcolor{black}{\small{}$\text{\ensuremath{\sqrt{AnnRV}}}$}} & \textcolor{black}{\small{}8.757} & \textcolor{black}{\small{}3.082} & \textcolor{black}{\small{}16.113} & \textcolor{black}{\small{}0.998} & \textcolor{black}{\small{}0.924} & \textcolor{black}{\small{}0.771}\tabularnewline
\multicolumn{2}{c}{\textcolor{black}{\small{}$Ret(\%)$}} & \textcolor{black}{\small{}1.094} & \textcolor{black}{\small{}-0.410} & \textcolor{black}{\small{}15.224} & \textcolor{black}{\small{}-0.091} & \textcolor{black}{\small{}0.030} & \textcolor{black}{\small{}0.039}\tabularnewline
\multicolumn{2}{c}{\textcolor{black}{\small{}$VRP$}} & \textcolor{black}{\small{}3.300} & \textcolor{black}{\small{}-0.214} & \textcolor{black}{\small{}10.088} & \textcolor{black}{\small{}0.860} & \textcolor{black}{\small{}0.286} & \textcolor{black}{\small{}0.030}\tabularnewline
\bottomrule
\end{tabular}
\par\end{centering}
\textcolor{black}{\footnotesize{}Note: Variables are measure in percent
of annualized volatility. For instance, $\sqrt{AnnRV}=100\sqrt{\tfrac{252}{22}\sum_{i=1}^{22}\mathrm{RVcc}_{t-i+1}}$.}{\footnotesize\par}
\end{table}

We present summary statistics for the full sample period in Table
\ref{Summary_of_statistics}. The data consists of daily returns,
the daily realized variances (measured in units of annualized standard
deviation), the CBOE VIX, and the VRP. A number of interesting observation
can be made from Table \ref{Summary_of_statistics}. First, the distribution
of VIX is skewed to the right with one (or more) extremely large values,
and the same is seen for the realized volatility. Second, both time
series of volatility are highly persistent with large and slowly decaying
autocorrelations. Third, the VRP also has a large first-order autocorrelation
but its higher-order autocorrelations decay much faster than is the
case for the VIX and the realized variance. This suggest that the
two variables that the VRP is composed of, have a common stochastic
trend that cancels out in the difference between the two variables
(a type of ``cointegration''). On average, the VIX is larger than
the realized volatility, with the average VRP being around 4.8\%.
\begin{table}[!tbh]
\begin{centering}
{\small{}\caption{Parameter estimates (full sample){\label{tab:parameters}}}
}{\small\par}
\par\end{centering}
{\small{}\medskip{}
}{\small\par}
\begin{centering}
\begin{tabular}{c>{\centering}m{2.2cm}>{\centering}m{2.2cm}>{\centering}m{2.2cm}>{\centering}m{2.2cm}>{\centering}m{2.2cm}}
\toprule 
\textcolor{black}{\small{}Model} & {\small{}$\mathrm{RG}$} & {\small{}$\mathrm{EG}$} & {\small{}$\mathrm{G}$} & {\small{}$\mathrm{HN}$} & {\small{}$\mathrm{HN}_{\mathtt{vd}}$}\tabularnewline
\midrule 
\textcolor{black}{\small{}$\lambda$} & {\small{}0.015} & {\small{}0.153} & {\small{}0.305} & \textcolor{black}{\small{}4.518} & \textcolor{black}{\small{}9.128}\tabularnewline
 & {\small{}(0.010)} & {\small{}(0.004)} & {\small{}(0.018)} & \textcolor{black}{\small{}(0.259)} & \textcolor{black}{\small{}(1.635)}\tabularnewline
\textcolor{black}{\small{}$\omega$} & {\small{}-0.088} & {\small{}-0.086} & {\small{}1.60E-06} & \textcolor{black}{\small{}-1.44E-06} & \textcolor{black}{\small{}-1.39E-06}\tabularnewline
 & {\small{}(0.014)} & {\small{}(0.002)} & {\small{}(1.09E-07)} & \textcolor{black}{\small{}(9.27E-08)} & \textcolor{black}{\small{}(1.13E-07)}\tabularnewline
\textcolor{black}{\small{}$\beta$} & {\small{}0.991} & {\small{}0.990} & {\small{}0.940} & \textcolor{black}{\small{}0.870} & \textcolor{black}{\small{}0.895}\tabularnewline
 & {\small{}(0.001)} & {\small{}2.22E-04} & {\small{}(0.004)} & \textcolor{black}{\small{}(0.010)} & \textcolor{black}{\small{}(0.009)}\tabularnewline
\textcolor{black}{\small{}$\alpha$} &  &  & {\small{}0.054} & \textcolor{black}{\small{}3.10E-06} & \textcolor{black}{\small{}2.24E-06}\tabularnewline
 &  &  & {\small{}(0.004)} & \textcolor{black}{\small{}(1.50E-07)} & \textcolor{black}{\small{}(2.16E-07)}\tabularnewline
\textcolor{black}{\small{}$\delta$} &  &  &  & \textcolor{black}{\small{}197.183} & \textcolor{black}{\small{}202.167}\tabularnewline
 &  &  &  & \textcolor{black}{\small{}(12.852)} & \textcolor{black}{\small{}(18.782)}\tabularnewline
\textcolor{black}{\small{}$\tau_{1}$} & {\small{}-0.073} & {\small{}-0.062} &  &  & \tabularnewline
 & {\small{}(0.005)} & {\small{}(0.002)} &  &  & \tabularnewline
\textcolor{black}{\small{}$\tau_{2}$} & {\small{}0.012} & {\small{}0.096} &  &  & \tabularnewline
 & {\small{}(0.002)} & {\small{}(0.001)} &  &  & \tabularnewline
\textcolor{black}{\small{}$\gamma$} & {\small{}0.080} &  &  &  & \tabularnewline
 & {\small{}(0.009)} &  &  &  & \tabularnewline
\textcolor{black}{\small{}$\kappa$} & {\small{}0.427} &  &  &  & \tabularnewline
 & {\small{}(0.278)} &  &  &  & \tabularnewline
\textcolor{black}{\small{}$\phi$} & {\small{}1.078} &  &  &  & \tabularnewline
 & {\small{}(0.029)} &  &  &  & \tabularnewline
\textcolor{black}{\small{}$\delta_{1}$} & {\small{}-0.083} &  &  &  & \tabularnewline
 & {\small{}(0.010)} &  &  &  & \tabularnewline
\textcolor{black}{\small{}$\delta_{2}$} & {\small{}0.129} &  &  &  & \tabularnewline
 & {\small{}(0.010)} &  &  &  & \tabularnewline
\textcolor{black}{\small{}$\sigma^{2}$} & {\small{}0.325} &  &  &  & \tabularnewline
 & {\small{}(0.010)} &  &  &  & \tabularnewline
\textcolor{black}{\small{}$\xi$} & {\small{}-1.07} &  &  &  & \tabularnewline
 & {\small{}(0.130)} &  &  &  & \tabularnewline
{\small{}$\eta$} &  &  &  &  & {\small{}1.143}\tabularnewline
 &  &  &  &  & {\small{}(0.061)}\tabularnewline
\textcolor{black}{\small{}$\pi^{\mathbb{P}}$} & 0.991 & 0.990 & 0.993 & 0.990 & 0.986\tabularnewline
{\small{}$\ell_{r}$} & {\small{}12863.96} & 12661.65 & 12481.99 & 12610.18 & 12693.12\tabularnewline
{\small{}$\ell_{x}$} & {\small{}-3229.12} &  &  &  & \tabularnewline
\textcolor{black}{\small{}$\ell_{\mathtt{vix}}$} & {\small{}-8811.76} & -9362.61 & -9509.19 & -10144.56 & -10079.12\tabularnewline
\textcolor{black}{\small{}$\ell_{r,x}$} & {\small{}9634.84} &  &  &  & \tabularnewline
\textcolor{black}{\small{}$\ell_{r,\mathtt{vix}}$} & {\small{}4052.202} & 3299.045 & 2972.800 & 2465.615 & 2613.997\tabularnewline
\textcolor{black}{\small{}$\ell_{r,x,\mathtt{vix}}$} & 823.083 &  &  &  & \tabularnewline
\bottomrule
\end{tabular}
\par\end{centering}
{\small{}\medskip{}
}\textcolor{black}{\small{}Note: Parameter estimates are reported
with robust standard errors in parenthesis. }{\small{}The persistence
parameter }\textcolor{black}{\small{}$\pi^{\mathbb{P}}=\beta$}{\small{}
is in the $\mathrm{RG}$ and $\mathrm{EG}$ models, $\pi^{\mathbb{P}}=\alpha+\beta$
in the model}\textcolor{black}{\small{}, and $\pi^{\mathbb{P}}=\beta+\alpha\delta^{2}$
in the two Heston-Nandi models. For the }{\small{}$\mathrm{HN}_{\mathtt{vd}}$}\textcolor{black}{\small{}
model we report , $\eta=(1+2\alpha\xi)^{-1}$ (which is the the variance
risk ratio $h_{t}^{*}/h_{t}$) instead of $\xi$. (The implied value
for $\xi$ is here -55,768.95).}{\small\par}
\end{table}

\subsection{Model Estimates (Full Sample)}

In the following, we use the following abbreviations for the models:
$\mathrm{RG}$ for the Realized GARCH model, $\mathrm{EG}$ for EGARCH,
$\mathrm{G}$ for GARCH, $\mathrm{HN}$ for Heston-Nandi GARCH, and
$\mathrm{HN}_{\mathtt{vd}}$ for Heston-Nandi GARCH with the variance
dependent SDF.

We present the parameter estimates for the full sample period for
each of the five models in Table \ref{tab:parameters} along with
robust standard errors in parentheses and some additional statistics. 

An interesting observation can be made about the market price of equity
risk, $\lambda$. This parameter is similar for the first three models,
however the estimated of $\lambda$ in the EGARCH and GARCH models
are 10-20 times larger than the estimate for the Realized GARCH model.
The estimates of $\lambda$ in the EGARCH and GARCH models are in
line with those reported in \citet{HaoZhang2013}. If the model are
estimated solely from return data, then the estimates of $\lambda$
are much smaller, see \citet{HaoZhang2013}. The reason is that the
$\mathrm{EG}$ and $\mathrm{G}$ models lack a separate volatility
risk parameters, and the models inflate the value of $\lambda$ in
order to compensate for the volatility risk that is embedded in the
VIX. The $\lambda$ for the Heston-Nandi model is not directly comparable
to those of the other models, because this coefficient is associated
with $h$ in the return equation, rather than $\sqrt{h}$ (for the
other models).

The persistence parameters under the $\mathbb{P}$-measure is denoted
$\pi^{\mathbb{P}}$ and is defined by $\beta$ for RG and EG, by $\alpha+\beta$
for G, and by $\beta+\alpha\delta^{2}$ for $\mathrm{HN}$ and $\mathrm{HN}_{\mathtt{vd}}$.
The persistence is quite similar across models and close to unity
in all cases. The estimates of $\tau_{1}$ and $\delta_{1}$ are negative
for both RG model and the EG model, which reflect a negative correlation
between return and volatility shocks. This is the so-called leverage
effect and these findings are consistent with the existing literature.

The estimate of the volatility risk parameter in the RG model, $\xi$,
is negative and significant. From the decomposition of the (log) VRP
we can compute the relative contributions of the two terms in (\ref{eq:vrp_decom})
using the estimated RG model. The first term is compensation for the
equity risk premium and its contribution ($\propto-\tau_{1}\lambda+\tau_{2}\lambda^{2}$)
is estimated to be 2.2\%. The second term is the separate compensation
for volatility risk and its contribution ($\propto-\gamma\sigma\xi$)
is estimated to be 97.8\%. This suggests that the majority of VRP
is due to compensation for the volatility shock, $u_{t}$, and only
a small of fraction of the VRP can be attributed to the leverage effect
and the equity premium.\footnote{This finding is specific to the RG model structure, that only includes
a short-term leverage effect. So it is possible that that models with
a more sophisticated leverage effect and/or long memory feature, would
result in different weights on the two terms.} This empirical finding supports the view in \citet{HaoZhang2013}
who argued that equity risk cannot justify the observed market VRP.

The value of the maximized log-likelihood function is a measure of
the model's ability to fit the empirical distribution of the observed
data. The Realized GARCH model with the affine exponential SDF clearly
has the best fit for all terms of the log-likelihood that are directly
comparable. Both $\ell_{r}$ and $\ell_{\mathrm{vix}}$ and their
sum $\ell_{r,\mathrm{vix}}$, are much larger for the Realized GARCH
model than any of the other models. This is despite the fact that
the other models seek to maximize $\ell_{r,\mathrm{vix}}$ while the
objective of the Realized GARCH model entails a tradeoff between this
term and the log-likelihood for the realized measures, $\ell_{x}$.
Following the Realized GARCH model, the $\mathrm{HN}_{\mathtt{vd}}$
has the second best performance in terms of describing returns, $\ell_{r}$,
whereas the EGARCH takes the second spot in terms of explaining the
variation in the VIX, $\ell_{\mathrm{vix}}$. Below, we evaluate the
model's ability to describe the VIX in greater details. 

\subsection{Model Performance's for VIX, VRP, and Volatility}

In this section, we focus on the models' ability to explain the variation
in the VIX, VRP, and the volatility of cumulative returns. First,
we report summary statistics for the full sample, then we report results
for various subsamples -- in-sample results as well as out-of-sample
results. Most of the existing literature has focused on a single variable.
For instance, the focus in \citet{HaoZhang2013}, \citet{Christoffersen2014},
and \citet{Corsi2015JOE} was VIX and derivative pricing, whereas
\citet{Wang2017} focused on volatility under the physical measure.

\subsubsection{Comparison of market and model-based VIX}

In this section, we evaluate the model's ability to describe the VIX
in greater details beyond the log-likelihood term, $\ell_{\mathrm{vix}}$,
listed above. Table \ref{tab:summary} reports a range of summary
statistics based on the full sample, where we compare the model-based
measures of VIX with the observed VIX.
\begin{table}[!tbh]
\caption{VIX pricing performance (full sample){\label{tab:summary}}}

\begin{centering}
\medskip{}
{\small{} }%
\begin{tabular}{r@{\extracolsep{0pt}.}lcccccc}
\toprule 
\multicolumn{2}{c}{{\small{}Model}} & {\small{}$\mathrm{RG}$} & {\small{}$\mathrm{EG}$} & {\small{}$\mathrm{G}$} & {\small{}$\mathrm{HN}$} & {\small{}$\mathrm{HN_{vd}}$} & \multicolumn{1}{c}{{\small{}VIX}}\tabularnewline
\midrule 
\multicolumn{2}{c}{{\small{}Bias}} & {\small{}-0.048} & {\small{}-0.203} & {\small{}-0.318} & {\small{}-0.203} & {\small{}0.164} & \tabularnewline
\multicolumn{2}{c}{{\small{}MAE}} & {\small{}1.866} & {\small{}2.222} & {\small{}2.199} & {\small{}2.439} & {\small{}2.419} & \tabularnewline
\multicolumn{2}{c}{{\small{}RMSE}} & {\small{}2.504} & {\small{}2.898} & {\small{}3.012} & {\small{}3.565} & {\small{}3.504} & \tabularnewline
\multicolumn{2}{c}{{\small{}Corr}} & {\small{}0.959} & {\small{}0.945} & {\small{}0.941} & {\small{}0.916} & {\small{}0.919} & \tabularnewline
\multicolumn{2}{c}{{\small{}AR1}} & {\small{}0.994} & {\small{}0.994} & {\small{}0.996} & {\small{}0.993} & {\small{}0.992} & {\small{}0.980}\tabularnewline
\multicolumn{2}{c}{{\small{}AR10}} & {\small{}0.901} & {\small{}0.935} & {\small{}0.943} & {\small{}0.933} & {\small{}0.926} & {\small{}0.898}\tabularnewline
\multicolumn{2}{c}{{\small{}AR22}} & {\small{}0.775} & {\small{}0.836} & {\small{}0.841} & {\small{}0.854} & {\small{}0.836} & {\small{}0.810}\tabularnewline
\multicolumn{2}{c}{{\small{}Mean}} & {\small{}18.362} & {\small{}18.207} & {\small{}18.093} & {\small{}18.207} & {\small{}18.574} & {\small{}18.410}\tabularnewline
\multicolumn{2}{c}{{\small{}Var}} & {\small{}73.173} & {\small{}71.775} & {\small{}71.594} & {\small{}59.951} & {\small{}58.190} & {\small{}77.652}\tabularnewline
\multicolumn{2}{c}{{\small{}Skew}} & {\small{}2.560} & {\small{}2.560} & {\small{}3.556} & {\small{}1.714} & {\small{}1.779} & {\small{}2.657}\tabularnewline
\multicolumn{2}{c}{{\small{}Kurt}} & {\small{}12.301} & {\small{}12.282} & {\small{}18.909} & {\small{}6.580} & {\small{}6.862} & {\small{}12.662}\tabularnewline
\bottomrule
\end{tabular}{\small\par}
\par\end{centering}
\medskip{}
{\small{}Note: Summary statistics for the VIX errors, $e_{t}=\mathrm{VIX}_{t}^{\mathrm{model}}-\mathrm{VIX}_{t}^{\mathrm{market}}$.
We report the the sample average of $e_{t}$ (Bias), the mean absolute
errors (MAE), the root of mean squared errors (RMSE), the sample correlation
between $\mathrm{VIX}_{t}^{\mathrm{model}}$ and $\mathrm{VIX}_{t}^{\mathrm{market}}$
(Corr), and the sample autocorrelations of $e_{t}$ for lags 1, 10,
and 22, that are denoted AR1, AR10, and AR22, respectively. For $\mathrm{VIX}_{t}^{\mathrm{model}}$
we report its sample average (Mean), its sample variance (Var), its
sample skewness, (Skew), and its sample excess kurtosis (Kurt).}{\small\par}
\end{table}

In this comparison, the Realized GARCH model is also consistently
the best model. It has the smallest bias, the smallest mean squared
error, and the smallest mean absolute error. The models: EGARCH, GARCH,
and Heston-Nandi with LRNVR tend to underestimate the VIX, whereas
the Heston-Nandi GARCH with the variance dependent SDF tends to overestimate
the VIX. The Realized GARCH model also has the highest correlation
between the model-implied VIX and the market-based VIX. With the Realized
GARCH model, the resulting statistical properties of the model-based
VIX are closer to those of the market-based VIX, that those of other
models.

\subsubsection{In-Sample Comparison of the VRP and and Its Components}

Table \ref{tab:In-sample-fit} provides the in-sample pricing performance
for the variance risk premium and its two components: the volatility
index (VIX) and the annualized model-based volatility, and we report
the bias for each of the models.\footnote{The annualized volatility is calculated based on the martingale process
assumption made by \citet{BollerslevTauchenZhou2009} and our results
are robust when annualized volatility is calculated based on the forecast
value using HAR model (method used in \citet{Bekaert2014} etc.).
See section \ref{Sec:RobustalterVRP} for details.}

In terms of the volatility risk premium, the Realized GARCH model
provides the smallest bias, the smallest root mean square error (RMSE),
and the smallest mean absolute error (MAE). The reduction in pricing
errors relative to other models ranges from 15.0\% to 30.6\% in terms
of RMSE and 27.1\% to 44.7\% in terms of MAE. Among the competing
models, the Heston-Nandi GARCH model with variance dependent SDF appears
to be the best alternative. In contrast, the the Heston-Nandi model
with LRNVR, which is arguably a very popular option pricing model,
does not fair well in terms of explaining the volatility risk premium.
The EGARCH model performs significantly better than other GARCH models
using LRNVR, especially in terms of the RMSE.

We observe very similar patterns across models in terms of their ability
to price the VIX. The Realized GARCH delivers the best performance
while the Heston-Nandi GARCH takes last place. In fact, the non-affine
models (RG, EG and G) perform substantially better that the two affine
models ($\mathrm{HN}$ and $\mathrm{HN}_{\mathtt{vd}}$), which is
consistent with the existing literature on option pricing with GARCH
models. The main advantage of the affine models is their analytical
expressions for the moment generating function. Fortunately, these
are not needed for VIX pricing, so the non-affine model are clearly
preferred for this problem.\footnote{It is very difficult, if not impossible, to obtain an analytical moment
generation function for cumulative returns or the $k$-step ahead
conditional volatility for these non-affine models. Both are needed
for quasi-analytical pricing formula for derivatives using a Fourier
inverse transformation. For this reason, computationally intensive
simulation methods and analytical expansions are commonly used for
pricing derivatives with non-affine models, see \citet{RGOP2017}
for the use of an Edgeworth expansion to price options with the Realized
GARCH model.} The performance gain for the Realized GARCH model ranges from 13.6\%
to 29.8\% in terms of RMSE and between 15.1\% and 23.5\% in terms
of MAE.

All models tend to over-estimate the expected volatility under the
physical measure. However, the bias is much smaller for the Realized
GARCH model. This indicates that the other models, in order to price
the VIX, inadvertently increase the level of volatility to compensate
for their shortcomings in risk neutralization. The RG and $\mathrm{HN}_{\mathtt{vd}}$
both have additional parameter to compensate for volatility risk,
which likely explain their smaller bias. In terms of explaining the
annualized volatility, the Realized GARCH model reduces the RMSE by
32.8\% to 48.1\% and the MAE is reduced by 48.6\% to 58.6\%.

It is worth emphasizing that the parameter estimation does not target
the volatility risk premium directly, the superior performance of
Realized GARCH highlights the model's ability of reconcile the physical
and risk neutral dynamics within a single model framework. This is
some accomplishment by the Realized GARCH framework, because this
was considered to be a very difficult empirical problem, see \citet{Bates1996}.
\begin{table}[!tbh]
\caption{In-sample Statistics of Model Performance\label{tab:In-sample-fit}}

\medskip{}

\begin{centering}
\begin{tabular}{cccccc}
\toprule 
\textbf{\small{}Model} & \textbf{\small{}$\mathrm{RG}$} & \textbf{\small{}$\mathrm{EG}$} & \textbf{\small{}$\mathrm{G}$} & \textbf{\small{}$\mathrm{HN}$} & \textbf{\small{}$\mathrm{HN}_{\mathtt{vd}}$}\tabularnewline
\midrule 
\multicolumn{3}{l}{\textbf{\small{}Volatility Risk Premium}} &  &  & \tabularnewline
\midrule 
{\small{}Bias} & {\small{}-0.497} & {\small{}-3.827} & {\small{}-4.397} & {\small{}-4.350} & {\small{}-3.053}\tabularnewline
{\small{}RMSE} & {\small{}3.825} & {\small{}5.078} & {\small{}5.509} & {\small{}5.456} & {\small{}4.499}\tabularnewline
{\small{}$\triangle\%$} &  & {\small{}24.7\%} & {\small{}30.6\%} & {\small{}29.9\%} & {\small{}15.0\%}\tabularnewline
{\small{}MAE} & {\small{}2.635} & {\small{}4.316} & {\small{}4.766} & {\small{}4.720} & {\small{}3.612}\tabularnewline
{\small{}$\triangle\%$} &  & {\small{}39.0\%} & {\small{}44.7\%} & {\small{}44.2\%} & {\small{}27.1\%}\tabularnewline
 &  &  &  &  & \tabularnewline
\multicolumn{3}{l}{\textbf{\small{}Volatility Index (VIX)}} &  &  & \tabularnewline
\midrule 
{\small{}Bias} & {\small{}-0.048} & {\small{}-0.203} & {\small{}-0.318} & {\small{}-0.203} & {\small{}0.164}\tabularnewline
{\small{}RMSE} & {\small{}2.504} & {\small{}2.898} & {\small{}3.012} & {\small{}3.565} & {\small{}3.504}\tabularnewline
{\small{}$\triangle\%$} &  & {\small{}13.6\%} & {\small{}16.9\%} & {\small{}29.8\%} & {\small{}28.5\%}\tabularnewline
{\small{}MAE} & {\small{}1.866} & {\small{}2.222} & {\small{}2.199} & {\small{}2.439} & {\small{}2.419}\tabularnewline
{\small{}$\triangle\%$} &  & {\small{}16.0\%} & {\small{}15.1\%} & {\small{}23.5\%} & {\small{}22.9\%}\tabularnewline
 &  &  &  &  & \tabularnewline
\multicolumn{3}{l}{\textbf{\small{}Annualized Volatility}} &  &  & \tabularnewline
\midrule 
{\small{}Bias} & {\small{}0.449} & {\small{}3.624} & {\small{}4.080} & {\small{}4.146} & {\small{}3.217}\tabularnewline
{\small{}RMSE} & {\small{}2.943} & {\small{}4.380} & {\small{}4.658} & {\small{}5.666} & {\small{}4.920}\tabularnewline
{\small{}$\triangle\%$} &  & {\small{}32.8\%} & {\small{}36.8\%} & {\small{}48.1\%} & {\small{}40.2\%}\tabularnewline
{\small{}MAE} & {\small{}2.002} & {\small{}3.898} & {\small{}4.285} & {\small{}4.839} & {\small{}3.979}\tabularnewline
{\small{}$\triangle\%$} &  & {\small{}48.6\%} & {\small{}53.3\%} & {\small{}58.6\%} & {\small{}49.7\%}\tabularnewline
\bottomrule
\end{tabular}
\par\end{centering}
Note: Bias denotes the difference between the model-implied quantity
and the market-based quantity. The rows indicated with ``$\triangle\%$''
present the increase in RMSE and MAE for all models relative to the
$\mathrm{RG}$. The market-based VRP in this table is defined by $\mathrm{VRP}_{t}^{\mathrm{market}}=\mathrm{VIX}_{t}-\sqrt{\tfrac{252}{22}\sum_{i=1}^{22}\mathrm{RVcc}_{t-i+1}}\times100$.
\end{table}

\begin{center}
 
\par\end{center}

\subsubsection{Out-of-sample pricing performance}

The proposed pricing model, which is based on the Realized GARCH model
and the affine exponential SDF, requires a larger number of parameters
to be estimated than the methods based on the conventional GARCH models.
The larger number of parameters could entail some overfitting of the
model, and this might explain some of the observed empirical improvements.
It is therefore important to document that the model also provides
improvements out-of-sample.\textcolor{black}{{} In this section, we
compare the models in terms of their out-of-sample pricing errors
using a rolling estimation window, based on the past 750 daily observations.
The first forecast is made for the first month (22 trading days) of
2007, and this forecasts is based on parameters that were estimated
with the previous 750 daily observations (January 6th 2004 to December
29th, 2006). We report out-of-sample pricing errors for 2007-2018
}and two sub-sample periods, 2007-2012 and 2013-2018. Splitting the
out-of-sample period into two subsamples is interesting because some
results could potentially be specific to the global financial crisis,
which had high volatility and high volatility-of-volatility. The global
financial crisis is contained in the first subsample.

We report a range of performance statistics for each of the models
and each of the sample periods in Table \ref{tab:OOS_TPT}. The significance
of relative performance is evaluated with Diebold-Mariano (DM) statistics
where we compare each of the alternative models to the Realized GARCH
model. For this purpose, we first compute the tracking errors for
the $i$-th model,

\[
e_{it}=X_{i,t}^{\mathtt{model}}-X_{t}^{\mathtt{market}},
\]
where $X$ represents the volatility risk premium, the volatility
index, or the annualized volatility. These errors are translated into
losses using either the mean square error, $g(e_{it})=e_{it}^{2}$
or the mean absolute error $g(e_{it})=|e_{it}|$. The loss of model
$i$, relative to the Realized GARCH model ($i=0$) is now defined
by

\[
d_{i,t}=g(e_{it})-g(e_{0t}),
\]
and we proceed to tests the hypothesis, $H_{0}:\thinspace\thinspace\mathbb{E}(d_{i,t})=0$
using the Diebold-Mariano (DM) statistic, $\mathrm{DM}_{i}=\sqrt{T}\bar{d}_{i,\cdot}/\hat{\sigma}_{d_{i}}$,
where $\bar{d}_{i,\cdot}=\frac{1}{T}\sum_{t}d_{i,t}$ and $\hat{\sigma}_{d_{i}}^{2}$
is an estimate of the long-run variance of $\{d_{i,t}\}$. Our estimates
of $\sigma_{d_{i}}^{2}$, are based on the Parzen kernel with bandwidth
$H=42$. Under suitable regularity conditions, it can be shown that
$\mathrm{DM}_{i}\overset{d}{\rightarrow}N(0,1)$ under the null hypothesis
$\mathbb{E}(d_{i,t})=0$, see \citet{DieboldMariano95}, and the 10\%
and 5\% critical values are therefore given by 1.64 and 1.96, respectively.

Once again the Realized GARCH provides the best out-of-sample pricing
performance for the VRP, the VIX, and the annualized volatility, and
this is found in all three sample periods. The model also provides
the smallest bias in most cases and the improvement in RMSE/MAE ranges
from 10\% to 40\% in most cases, which is significant in most cases.
This shows that superior in-sample performance of the Realized GARCH
model cannot be attributed to overfitting. The RMSE and MAE are, as
expected, larger in the subsample with the financial crisis. The Realized
GARCH model really stands out in terms of explaining the volatility
under the physical measure (Annualized Volatility), where the reduction
in out-of-sample loss is always larger than 20\% and as larger as
52.2\%. 

The HN is always the worst model for tracking the VIX out-of-sample,
as was the case in-sample. The picture is similar for the annualized
volatility, albeit $\mathrm{HN}_{\mathtt{vd}}$ is the ``best of
the rest'' in terms of tracking of annualized volatility in the post
crisis period. This might be explained by the variance dependent SDF
being less misspecified when volatility is low, while it cannot generate
enough discrepancy between $\mathbb{P}$ and $\mathbb{Q}$ when volatility
is high. Interestingly, even though the $\mathrm{HN}_{\mathtt{vd}}$
is clearly inferior to both the GARCH and EGARCH models in terms of
forecasting volatility under $\mathbb{P}$ and $\mathbb{Q}$, it is
actually better than these two models in terms of forecasting the
VRP. 
\begin{sidewaystable}
\begin{centering}
\centering \caption{Out-of-sample Statistics of Model Performance\label{tab:OOS_TPT}}
\begin{tabular}{cccccccccccccccc}
\toprule 
 & \multicolumn{5}{c}{\textbf{\small{}Volatility Risk Premium}} & \multicolumn{5}{c}{\textbf{\small{}Volatility Index (VIX)}} & \multicolumn{5}{c}{\textbf{\small{}Annualized Volatility}}\tabularnewline
\midrule 
 &  &  &  &  &  &  &  &  &  &  &  &  &  &  & \tabularnewline
\multicolumn{16}{l}{\textbf{\small{}Full: 20070103 - 20181230}}\tabularnewline
\textbf{\small{}Model} & \textbf{\small{}$\mathrm{RG}$} & \textbf{\small{}$\mathrm{EG}$} & \textbf{\small{}$\mathrm{G}$} & \textbf{\small{}$\mathrm{HN}$} & \textbf{\small{}$\mathrm{HN}_{\mathtt{vd}}$} & \textbf{\small{}$\mathrm{RG}$} & \textbf{\small{}$\mathrm{EG}$} & \textbf{\small{}$\mathrm{G}$} & \textbf{\small{}$\mathrm{HN}$} & \textbf{\small{}$\mathrm{HN}_{\mathtt{vd}}$} & \textbf{\small{}$\mathrm{RG}$} & \textbf{\small{}$\mathrm{EG}$} & \textbf{\small{}$\mathrm{G}$} & \textbf{\small{}$\mathrm{HN}$} & \textbf{\small{}$\mathrm{HN}_{\mathtt{vd}}$}\tabularnewline
\textbf{\small{}Bias} & {\small{}-0.688} & {\small{}-3.824} & {\small{}-3.762} & {\small{}-3.814} & {\small{}-2.753} & {\small{}0.155} & {\small{}-0.212} & {\small{}0.063} & {\small{}-0.533} & {\small{}0.132} & {\small{}0.844} & {\small{}3.612} & {\small{}3.825} & {\small{}3.282} & {\small{}2.885}\tabularnewline
\textbf{\small{}RMSE} & {\small{}3.975} & {\small{}5.420} & {\small{}5.378} & {\small{}5.387} & {\small{}4.668} & {\small{}2.870} & {\small{}3.259} & {\small{}3.278} & {\small{}4.298} & {\small{}3.903} & {\small{}3.271} & {\small{}5.154} & {\small{}4.961} & {\small{}5.856} & {\small{}5.360}\tabularnewline
\textbf{\small{}DM stat.} &  & {\small{}8.266} & {\small{}8.064} & {\small{}8.386} & {\small{}4.853} &  & {\small{}3.976} & {\small{}3.579} & {\small{}2.659} & {\small{}1.887} &  & {\small{}6.609} & {\small{}5.400} & {\small{}5.314} & {\small{}4.066}\tabularnewline
\textbf{\small{}$\triangle\%$} &  & {\small{}26.7\%} & {\small{}26.1\%} & {\small{}26.2\%} & {\small{}14.8\%} &  & {\small{}11.9\%} & {\small{}12.4\%} & {\small{}33.2\%} & {\small{}26.5\%} &  & {\small{}36.5\%} & {\small{}34.1\%} & {\small{}44.2\%} & {\small{}39.0\%}\tabularnewline
\textbf{\small{}MAE} & {\small{}2.921} & {\small{}4.565} & {\small{}4.512} & {\small{}4.543} & {\small{}3.676} & {\small{}2.113} & {\small{}2.402} & {\small{}2.387} & {\small{}2.842} & {\small{}2.615} & {\small{}2.466} & {\small{}4.420} & {\small{}4.422} & {\small{}4.750} & {\small{}4.203}\tabularnewline
\textbf{\small{}DM stat.} &  & {\small{}10.179} & {\small{}10.170} & {\small{}10.621} & {\small{}6.250} &  & {\small{}3.995} & {\small{}3.333} & {\small{}4.222} & {\small{}3.137} &  & {\small{}9.034} & {\small{}8.894} & {\small{}11.298} & {\small{}8.624}\tabularnewline
\textbf{\small{}$\triangle\%$} &  & {\small{}36.0\%} & {\small{}35.3\%} & {\small{}35.7\%} & {\small{}20.5\%} &  & {\small{}12.0\%} & {\small{}11.5\%} & {\small{}25.6\%} & {\small{}19.2\%} &  & {\small{}44.2\%} & {\small{}44.2\%} & {\small{}48.1\%} & {\small{}41.3\%}\tabularnewline
 &  &  &  &  &  &  &  &  &  &  &  &  &  &  & \tabularnewline
\multicolumn{16}{l}{\textbf{\small{}Crisis Period: 20070103 - 20121230}}\tabularnewline
\textbf{\small{}Model} & \textbf{\small{}$\mathrm{RG}$} & \textbf{\small{}$\mathrm{EG}$} & \textbf{\small{}$\mathrm{G}$} & \textbf{\small{}$\mathrm{HN}$} & \textbf{\small{}$\mathrm{HN}_{\mathtt{vd}}$} & \textbf{\small{}$\mathrm{RG}$} & \textbf{\small{}$\mathrm{EG}$} & \textbf{\small{}$\mathrm{G}$} & \textbf{\small{}$\mathrm{HN}$} & \textbf{\small{}$\mathrm{HN}_{\mathtt{vd}}$} & \textbf{\small{}$\mathrm{RG}$} & \textbf{\small{}$\mathrm{EG}$} & \textbf{\small{}$\mathrm{G}$} & \textbf{\small{}$\mathrm{HN}$} & \textbf{\small{}$\mathrm{HN}_{\mathtt{vd}}$}\tabularnewline
\textbf{\small{}Bias} & {\small{}-1.011} & {\small{}-4.311} & {\small{}-4.439} & {\small{}-4.315} & {\small{}-3.376} & {\small{}-0.294} & {\small{}-0.824} & {\small{}-0.484} & {\small{}-1.427} & {\small{}-0.321} & {\small{}0.717} & {\small{}3.487} & {\small{}3.956} & {\small{}2.888} & {\small{}3.055}\tabularnewline
\textbf{\small{}RMSE} & {\small{}4.570} & {\small{}6.060} & {\small{}6.187} & {\small{}6.014} & {\small{}5.339} & {\small{}3.383} & {\small{}3.972} & {\small{}3.944} & {\small{}5.497} & {\small{}4.866} & {\small{}3.493} & {\small{}5.520} & {\small{}5.158} & {\small{}6.844} & {\small{}6.471}\tabularnewline
\textbf{\small{}DM stat.} &  & {\small{}6.012} & {\small{}6.489} & {\small{}5.940} & {\small{}3.756} &  & {\small{}4.028} & {\small{}3.277} & {\small{}2.515} & {\small{}1.676} &  & {\small{}4.307} & {\small{}3.018} & {\small{}4.141} & {\small{}3.540}\tabularnewline
\textbf{\small{}$\triangle\%$} &  & {\small{}24.6\%} & {\small{}26.1\%} & {\small{}24.0\%} & {\small{}14.4\%} &  & {\small{}14.8\%} & {\small{}14.2\%} & {\small{}38.5\%} & {\small{}30.5\%} &  & {\small{}36.7\%} & {\small{}32.3\%} & {\small{}49.0\%} & {\small{}46.0\%}\tabularnewline
\textbf{\small{}MAE} & {\small{}3.420} & {\small{}5.164} & {\small{}5.261} & {\small{}5.142} & {\small{}4.301} & {\small{}2.557} & {\small{}3.069} & {\small{}2.946} & {\small{}3.798} & {\small{}3.327} & {\small{}2.514} & {\small{}4.505} & {\small{}4.476} & {\small{}5.261} & {\small{}4.992}\tabularnewline
\textbf{\small{}DM stat.} &  & {\small{}7.760} & {\small{}8.200} & {\small{}7.997} & {\small{}5.092} &  & {\small{}4.254} & {\small{}2.820} & {\small{}3.906} & {\small{}2.541} &  & {\small{}5.559} & {\small{}5.287} & {\small{}8.462} & {\small{}7.740}\tabularnewline
\textbf{\small{}$\triangle\%$} &  & {\small{}33.8\%} & {\small{}35.0\%} & {\small{}33.5\%} & {\small{}20.5\%} &  & {\small{}16.7\%} & {\small{}13.2\%} & {\small{}32.7\%} & {\small{}23.1\%} &  & {\small{}44.2\%} & {\small{}43.8\%} & {\small{}52.2\%} & {\small{}49.6\%}\tabularnewline
 &  &  &  &  &  &  &  &  &  &  &  &  &  &  & \tabularnewline
\multicolumn{16}{l}{\textbf{\small{}Post-crisis Period: 20130103-20181230}}\tabularnewline
\textbf{\small{}Model} & \textbf{\small{}$\mathrm{RG}$} & \textbf{\small{}$\mathrm{EG}$} & \textbf{\small{}$\mathrm{G}$} & \textbf{\small{}$\mathrm{HN}$} & \textbf{\small{}$\mathrm{HN}_{\mathtt{vd}}$} & \textbf{\small{}$\mathrm{RG}$} & \textbf{\small{}$\mathrm{EG}$} & \textbf{\small{}$\mathrm{G}$} & \textbf{\small{}$\mathrm{HN}$} & \textbf{\small{}$\mathrm{HN}_{\mathtt{vd}}$} & \textbf{\small{}$\mathrm{RG}$} & \textbf{\small{}$\mathrm{EG}$} & \textbf{\small{}$\mathrm{G}$} & \textbf{\small{}$\mathrm{HN}$} & \textbf{\small{}$\mathrm{HN}_{\mathtt{vd}}$}\tabularnewline
\textbf{\small{}Bias} & {\small{}-0.365} & {\small{}-3.336} & {\small{}-3.084} & {\small{}-3.313} & {\small{}-2.129} & {\small{}0.605} & {\small{}0.401} & {\small{}0.611} & {\small{}0.362} & {\small{}0.585} & {\small{}0.970} & {\small{}3.737} & {\small{}3.695} & {\small{}3.675} & {\small{}2.714}\tabularnewline
\textbf{\small{}RMSE} & {\small{}3.271} & {\small{}4.694} & {\small{}4.423} & {\small{}4.676} & {\small{}3.881} & {\small{}2.242} & {\small{}2.339} & {\small{}2.435} & {\small{}2.591} & {\small{}2.603} & {\small{}3.032} & {\small{}4.758} & {\small{}4.755} & {\small{}4.663} & {\small{}3.946}\tabularnewline
\textbf{\small{}DM stat.} &  & {\small{}5.903} & {\small{}5.400} & {\small{}6.254} & {\small{}3.188} &  & {\small{}1.256} & {\small{}1.704} & {\small{}2.879} & {\small{}2.846} &  & {\small{}6.162} & {\small{}6.981} & {\small{}6.516} & {\small{}4.014}\tabularnewline
\textbf{\small{}$\triangle\%$} &  & {\small{}30.3\%} & {\small{}26.0\%} & {\small{}30.0\%} & {\small{}15.7\%} &  & {\small{}4.2\%} & {\small{}7.9\%} & {\small{}13.5\%} & {\small{}13.9\%} &  & {\small{}36.3\%} & {\small{}36.2\%} & {\small{}35.0\%} & {\small{}23.2\%}\tabularnewline
\textbf{\small{}MAE} & {\small{}2.421} & {\small{}3.966} & {\small{}3.763} & {\small{}3.943} & {\small{}3.050} & {\small{}1.670} & {\small{}1.734} & {\small{}1.828} & {\small{}1.884} & {\small{}1.902} & {\small{}2.418} & {\small{}4.335} & {\small{}4.367} & {\small{}4.238} & {\small{}3.414}\tabularnewline
\textbf{\small{}DM stat.} &  & {\small{}6.644} & {\small{}6.276} & {\small{}7.005} & {\small{}3.767} &  & {\small{}0.984} & {\small{}1.821} & {\small{}2.679} & {\small{}2.749} &  & {\small{}7.965} & {\small{}8.344} & {\small{}8.247} & {\small{}5.231}\tabularnewline
\textbf{\small{}$\triangle\%$} &  & {\small{}39.0\%} & {\small{}35.7\%} & {\small{}38.6\%} & {\small{}20.6\%} &  & {\small{}3.7\%} & {\small{}8.7\%} & {\small{}11.4\%} & {\small{}12.2\%} &  & {\small{}44.2\%} & {\small{}44.6\%} & {\small{}42.9\%} & {\small{}29.2\%}\tabularnewline
\bottomrule
\end{tabular}
\par\end{centering}
Note: Let the ``model error'' be the difference between the model-implied
quantity and the market-based quantity. It sample average is denoted
Bias and the rows indicated with ``$\triangle\%$'' state who much
larger the RMSE or MAE is for each alternative models, measured relative
to $\mathrm{RG}$. Diebold and Mariano statistics (DM stat. ) are
computed for the relative MSE or MAE losses, where the standard errors
are calculated with the Parzen kernel with $H=42$ as bandwidth. The
market VRP in this table is defined as $\mathrm{VRP}_{t}^{\mathrm{market}}=\mathrm{VIX}_{t}-\sqrt{\tfrac{252}{22}\sum_{i=1}^{22}\mathrm{RVcc}_{t-i+1}}\times100$.
\end{sidewaystable}

\section{Conclusion}

We have developed a Realized GARCH model for the simultaneous modeling
of returns, the VIX, and the VRP, using an exponentially affine SDF
that takes advantage of the dual shock structure in the Realized GARCH
model. This framework has several attractive features. First, its
dual-shock structure lead to a distinct compensation for volatility
risk, which is empirically important. Second, it takes advantage of
the information contained in realized measures of volatility. Third,
it has a flexible leverage function that captures the empirically
important return-volatility dependence in a parsimonious manner. Fourth,
the model combined with the exponentially affine SDF, conveniently,
yields analytical formulae for the VIX and the volatility risk premium.
Fifth, the model is an observation-driven model, which makes estimation
straight forward. Sixth, its dynamic properties under the physical
and risk-neutral measures offer intuitive and theory-consistent explanations
for the excellent empirical performance offered by the Realized GARCH
model. 

\bibliographystyle{ecta}
\addcontentsline{toc}{section}{\refname}\bibliography{reference,prh}

\newpage{}

\appendix
\setcounter{lem}{0}\renewcommand{\thelem}{A.\arabic{lem}}

\section{Appendix of Proofs\label{sec:Appendix-of-Proofs}}
\begin{lem}
\label{lem:LemConstantLambda}Suppose
\end{lem}
\[
M_{t+1}=\frac{\exp(-\lambda_{t}z_{t+1}-\xi_{t}u_{t+1})}{\mathrm{\mathbb{E}}\exp(-\lambda_{t}z_{t+1}-\xi_{t}u_{t+1})}=\exp\left\{ -\lambda_{t}z_{t+1}-\xi_{t}u_{t+1}-\tfrac{1}{2}(\lambda_{t}^{2}+\xi_{t}^{2})\right\} ,
\]
then by non-arbitrage we have $\lambda_{t}=\lambda$.
\begin{proof}
The non-arbitrage condition is $\mathbb{E}_{t}^{\mathbb{Q}}(\exp(r_{t+1}))=\exp(r)$
and the result follows by 
\begin{eqnarray*}
\mathbb{E}_{t}^{\mathbb{Q}}(\exp(r_{t+1})) & = & \mathbb{E}_{t}(M_{t+1}\exp(r_{t+1}))\\
 & = & \mathbb{E}_{t}\exp\left\{ (\sqrt{h_{t+1}}-\lambda_{t})z_{t+1}-\xi_{t}u_{t+1}-\tfrac{1}{2}(\lambda_{t}^{2}+\xi_{t}^{2})+r+\lambda\sqrt{h_{t+1}}-\tfrac{1}{2}h_{t+1}\right\} \\
 & = & \mathbb{E}_{t}\exp\left\{ r+(\lambda-\lambda_{t})\sqrt{h_{t+1}}\right\} .
\end{eqnarray*}

In order to have the above equation hold for all $h_{t}$, we need
$\lambda_{t}=\lambda$.
\end{proof}
\begin{lem}
\label{LemRGunderQ}The Realized GARCH model defined by (\ref{eq:return.eq})-(\ref{eq:measurement.eq})
and the affine exponential SDF, implied the model (\ref{eq:returnQ})-(\ref{eq:measureQ})
under the risk neutral measure, $\mathbb{Q}$.
\end{lem}
\begin{proof}
Substituting $(z_{t}^{\ast}-\lambda,u_{t}^{\ast}-\xi)$ for $(z_{t},u_{t})$
immediately yields{\small{}
\begin{eqnarray*}
\log h_{t+1} & = & \omega+\beta\log h_{t}+\tau_{1}(z_{t}^{\ast}-\lambda)+\tau_{2}[(z_{t}^{\ast}-\lambda)^{2}-1]+\gamma\sigma(u_{t}^{\ast}-\xi)\\
 & = & (\omega-\tau_{1}\lambda+\tau_{2}\lambda^{2}-\gamma\sigma\xi)+\beta\log h_{t}+(\tau_{1}-2\tau_{2}\lambda)z_{t}^{\ast}+\tau_{2}(z_{t}^{\ast2}-1)+\gamma\sigma u_{t}^{\ast},\\
\log x_{t} & = & \kappa+\phi\log h_{t}+\delta_{1}(z_{t}^{\ast}-\lambda)+\delta_{2}[(z_{t}^{\ast}-\lambda)^{2}-1]+\sigma(u_{t}^{\ast}-\xi)\\
 & = & (\kappa-\delta_{1}\lambda+\delta_{2}\lambda^{2}-\sigma\xi)+\phi\log h_{t}+(\delta_{1}-2\delta_{2}\lambda)z_{t}^{\ast}+\delta_{2}(z_{t}^{\ast2}-1)+\sigma u_{t}^{\ast}.
\end{eqnarray*}
}{\small\par}
\end{proof}
\begin{lem}
\label{lem:Eexp(ax+bxx)}Suppose that $X\sim N(0,1)$ then for $b<\tfrac{1}{2}$
$\mathbb{E}\exp\{aX+bX^{2}\}=\tfrac{1}{\sqrt{1-2b}}\exp\{\tfrac{a^{2}/2}{1-2b}\}$ 
\end{lem}
\begin{proof}
We have 
\[
e^{ax+bx^{2}}e^{-\tfrac{x^{2}}{2}}=e^{ax-\tfrac{x^{2}}{2}(1-2b)}=e^{\tfrac{a}{\sqrt{1-2b}}u-\tfrac{u^{2}}{2}},
\]
where $u=\sqrt{1-2b}x$. So integration by substitution yields 
\begin{eqnarray*}
\mathbb{E}e^{aX+bX^{2}} & = & \int\tfrac{1}{\sqrt{2\pi}}e^{ax+bx^{2}}e^{-\tfrac{x^{2}}{2}}\mathrm{d}x=\int\tfrac{1}{\sqrt{2\pi}}e^{\tfrac{a}{\sqrt{1-2b}}u-\tfrac{u^{2}}{2}}\tfrac{1}{\sqrt{1-2b}}\mathrm{d}u\\
 & = & \tfrac{1}{\sqrt{1-2b}}\mathrm{E}e^{\tfrac{a}{\sqrt{1-2b}}U}=\tfrac{1}{\sqrt{1-2b}}e^{\tfrac{1}{2}\tfrac{a^{2}}{1-2b}}.
\end{eqnarray*}
\end{proof}
\noindent\textbf{Proof of Theorem \ref{theo:VIXpricingRG}.} From
the risk neutral dynamics, we have $\log h_{t+1}=\tilde{\omega}+\beta\log h_{t}+v_{t}^{\ast}$
where 
\begin{eqnarray*}
v_{t}^{\ast} & = & \tilde{\tau}_{1}z_{t}^{*}+\tau_{2}(z_{t}^{\ast2}-1)+\gamma\sigma u_{t}^{\ast},
\end{eqnarray*}
so that $\log h_{t+k}=\beta^{k-1}\log h_{t+1}+\sum_{i=0}^{k-2}\beta^{i}(v_{t+k-1-i}^{\ast}+\tilde{\omega})$.
It follows that 
\[
\mathbb{E}_{t}^{Q}[h_{t+k}]=\mathbb{E}_{t}^{Q}[\exp\{\beta^{k-1}\log h_{t+1}+\sum_{i=0}^{k-2}\beta^{i}(v_{t+k-1-i}^{\ast}+\tilde{\omega})\}]=h_{t+1}^{\beta^{k-1}}\prod_{i=0}^{k-2}F_{i},
\]
where $F_{i}=\mathbb{E}_{t}^{\mathbb{Q}}\left[e^{\beta^{i}(v_{t+k-1-i}^{\ast}+\tilde{\omega})}\right]$.
Using the expression for $v_{t}^{\ast}$, and applying Lemma \ref{lem:Eexp(ax+bxx)},
we have,{\footnotesize{}
\begin{eqnarray*}
F_{i} & = & \mathbb{E}_{t}^{\mathbb{Q}}\left[\exp\left\{ \beta^{i}(\tilde{\omega}-\tau_{2})+\beta^{i}\tilde{\tau}_{1}z+\beta^{i}\tau_{2}z^{2}+\beta^{i}\gamma\sigma u\right\} \right]\\
 & = & \exp\left\{ \beta^{i}(\tilde{\omega}-\tau_{2})\right\} \mathbb{E}_{t}^{\mathbb{Q}}\left[\exp\left\{ \beta^{i}\tilde{\tau}_{1}z+\beta^{i}\tau_{2}z^{2}\right\} \right]\mathbb{E}_{t}^{\mathbb{Q}}\left[\exp\left\{ \beta^{i}\gamma\sigma u\right\} \right]\\
 & = & \exp\left\{ \beta^{i}(\tilde{\omega}-\tau_{2})\right\} (1-2\beta^{i}\tau_{2})^{-1/2}\exp\{\tfrac{1}{2}\tfrac{\beta^{2i}\tilde{\tau}_{1}^{2}}{1-2\beta^{i}\tau_{2}}\}\exp\left\{ \tfrac{1}{2}\beta^{2i}\gamma^{2}\sigma^{2}\right\} \\
 & = & (1-2\beta^{i}\tau_{2})^{-1/2}\exp\left\{ \beta^{i}(\tilde{\omega}-\tau_{2})+\tfrac{1}{2}\beta^{2i}[\tfrac{\tilde{\tau}_{1}^{2}}{1-2\beta^{i}\tau_{2}}+\gamma^{2}\sigma^{2}]\right\} ,
\end{eqnarray*}
}where we suppressed the superscripts and subscripts on $z$ and $u$
to simplify the exposition. $\square$

\setcounter{table}{0}\renewcommand{\thetable}{B.\arabic{table}}

\section{Additional Empirical Results: Robustness Checks}

This appendix provides additional empirical results. These results
primarily serve to demonstrate that the main empirical results are
robust to a variety of alternative approaches to the empirical analysis,
such as alternative measures of the market VRP, alternative objective
functions, and different choice for the realized measures.

\subsection{Robustness Check: HAR type VRP results\label{Sec:RobustalterVRP}}

The volatility risk premium used in the the main body of the paper
used the realized volatility over the past months as the predictor
of volatility for the following month, which implicitly use a martingale
assumption for volatility. Here we consider an alternative regression-based
approach, which is known as the HAR model, see \citet{Corsi2009}.\footnote{For other ways other ways to construct the expected volatility, see
\citet{Bekaert2014}} With this approach the monthly volatility is regressed on lagged
daily, weekly, and monthly volatility,

\[
\mathrm{RVcc}_{t+1:t+22}=\beta_{0}+\beta_{d}\mathrm{RVcc}_{t}+\beta_{w}\mathrm{RVcc}_{t-4:t}+\beta_{m}\mathrm{RVcc}_{t-21:t}+\epsilon_{t+1:t+22},
\]
where $\mathrm{RVcc}_{a:b}=\frac{1}{b-a+1}\sum_{t=a}^{b}\mathrm{RVcc}_{t}$.
We also estimate this model using a rolling window with 750 observations
(and lagged an additional 22 days to avoid look-ahead bias in our
estimates).\footnote{For time $t$ the is estimated with data located in $[t-750,t-22]$,
otherwise the dependent variable will contain future realizations
of $\mathrm{RVcc}$.} Table \ref{tab:In-sample-fit-HARtype} and table \ref{tab:OOS_TPT-HARtype}
provide in-sample and out-of-sample results, and all main results
are identical to those obtained with the martingale definition of
the VRP. The Realized GARCH model continues to outperform all other
models.
\begin{table}[!tbh]
\caption{In-Sample Statistics of Model Performance with Alternative VRP Specification\label{tab:In-sample-fit-HARtype}}

\medskip{}

\begin{centering}
\begin{tabular}{crcccc}
\toprule 
\textbf{Model} & \textbf{$\mathrm{RG}$} & \textbf{$\mathrm{EG}$} & \textbf{$\mathrm{G}$} & \textbf{$\mathrm{HN}$} & \textbf{$\mathrm{HN}_{\mathtt{vd}}$}\tabularnewline
\midrule 
\multicolumn{3}{l}{\textbf{Volatility Risk Premium}} &  &  & \tabularnewline
\midrule 
Bias & \textcolor{black}{0.774} & \textcolor{black}{-2.556} & \textcolor{black}{-3.126} & \textcolor{black}{-3.078} & \textcolor{black}{-1.782}\tabularnewline
RMSE & \textcolor{black}{3.682} & \textcolor{black}{4.404} & \textcolor{black}{4.812} & \textcolor{black}{4.746} & \textcolor{black}{4.133}\tabularnewline
$\triangle\%$ &  & \textcolor{black}{16.4\%} & \textcolor{black}{23.5\%} & \textcolor{black}{22.4\%} & \textcolor{black}{10.9\%}\tabularnewline
MAE & \textcolor{black}{2.184} & \textcolor{black}{3.070} & \textcolor{black}{3.514} & \textcolor{black}{3.458} & \textcolor{black}{2.750}\tabularnewline
$\triangle\%$ & \multicolumn{1}{c}{} & \textcolor{black}{28.9\%} & \textcolor{black}{37.8\%} & \textcolor{black}{36.8\%} & \textcolor{black}{20.6\%}\tabularnewline
 & \multicolumn{1}{c}{} &  &  &  & \tabularnewline
\multicolumn{3}{l}{\textbf{Volatility Index (VIX)}} &  &  & \tabularnewline
\midrule 
Bias & -0.048 & -0.203 & -0.318 & -0.203 & 0.164\tabularnewline
RMSE & 2.504 & 2.898 & 3.012 & 3.565 & 3.504\tabularnewline
$\triangle\%$ &  & 13.6\% & 16.9\% & 29.8\% & 28.5\%\tabularnewline
MAE & 1.866 & 2.222 & 2.199 & 2.439 & 2.419\tabularnewline
$\triangle\%$ &  & 16.0\% & 15.1\% & 23.5\% & 22.9\%\tabularnewline
 & \multicolumn{1}{c}{} &  &  &  & \tabularnewline
\multicolumn{3}{l}{\textbf{Annualized Volatility}} &  &  & \tabularnewline
\midrule 
Bias & \textcolor{black}{-0.823} & \textcolor{black}{2.353} & \textcolor{black}{2.808} & \textcolor{black}{2.875} & \textcolor{black}{1.946}\tabularnewline
RMSE & \textcolor{black}{3.783} & \textcolor{black}{4.753} & \textcolor{black}{4.547} & \textcolor{black}{5.893} & \textcolor{black}{5.516}\tabularnewline
$\triangle\%$ &  & \textcolor{black}{20.4\%} & \textcolor{black}{16.8\%} & \textcolor{black}{35.8\%} & \textcolor{black}{31.4\%}\tabularnewline
MAE & \textcolor{black}{2.393} & \textcolor{black}{3.653} & \textcolor{black}{3.596} & \textcolor{black}{4.285} & \textcolor{black}{3.881}\tabularnewline
$\triangle\%$ &  & \textcolor{black}{34.5\%} & \textcolor{black}{33.4\%} & \textcolor{black}{44.1\%} & \textcolor{black}{38.3\%}\tabularnewline
\bottomrule
\end{tabular}
\par\end{centering}
Note: Bias is defined as model-based quantity minus the corresponding
market-based quantity; The rows starting with $\triangle\%$ present
the increase in percentage of RMSE and MAE for the competing models
relative to $\mathrm{RG}$. In this table, the market VRP is defined
by $\mathrm{VRP}_{t}^{\mathtt{market}}=\mathrm{VIX}_{t}-\sqrt{\mathrm{\widehat{\mathrm{RVcc}}}_{t+1:t+22}}\times100$,
where the annualized volatility, $\sqrt{\mathrm{\widehat{\mathrm{RVcc}}}_{t+1:t+22}}\times100$,
is the predicted value from the HAR model.
\end{table}
\begin{sidewaystable}
\begin{centering}
\centering \caption{Out-of-sample Statistics of Model Performance with Alternative VRP
Specification \label{tab:OOS_TPT-HARtype}}
\par\end{centering}
\begin{centering}
\begin{tabular}{cccccccccccccccc}
\toprule 
 & \multicolumn{5}{c}{\textbf{\small{}Volatility Risk Premium}} & \multicolumn{5}{c}{\textbf{\small{}Volatility Index (VIX)}} & \multicolumn{5}{c}{\textbf{\small{}Annualized Volatility}}\tabularnewline
\midrule 
 &  &  &  &  &  &  &  &  &  &  &  &  &  &  & \tabularnewline
\multicolumn{16}{l}{\textbf{\small{}Full Sample Period: 20070103 - 20181230}}\tabularnewline
\textbf{\textcolor{black}{\small{}Model}} & \textbf{\textcolor{black}{\small{}$\mathrm{RG}$}} & \textbf{\textcolor{black}{\small{}$\mathrm{EG}$}} & \textbf{\textcolor{black}{\small{}$\mathrm{G}$}} & \textbf{\textcolor{black}{\small{}$\mathrm{HN}$}} & \textbf{\textcolor{black}{\small{}$\mathrm{HN}_{\mathtt{vd}}$}} & \textbf{\textcolor{black}{\small{}$\mathrm{RG}$}} & \textbf{\textcolor{black}{\small{}$\mathrm{EG}$}} & \textbf{\textcolor{black}{\small{}$\mathrm{G}$}} & \textbf{\textcolor{black}{\small{}$\mathrm{HN}$}} & \textbf{\textcolor{black}{\small{}$\mathrm{HN}_{\mathtt{vd}}$}} & \textbf{\textcolor{black}{\small{}$\mathrm{RG}$}} & \textbf{\textcolor{black}{\small{}$\mathrm{EG}$}} & \textbf{\textcolor{black}{\small{}$\mathrm{G}$}} & \textbf{\textcolor{black}{\small{}$\mathrm{HN}$}} & \textbf{\textcolor{black}{\small{}$\mathrm{HN}_{\mathtt{vd}}$}}\tabularnewline
\textbf{\textcolor{black}{\small{}Bias}} & \textcolor{black}{\small{}0.450} & \textcolor{black}{\small{}-2.685} & \textcolor{black}{\small{}-2.624} & \textcolor{black}{\small{}-2.676} & \textcolor{black}{\small{}-1.614} & \textcolor{black}{\small{}0.155} & \textcolor{black}{\small{}-0.212} & \textcolor{black}{\small{}0.063} & \textcolor{black}{\small{}-0.533} & \textcolor{black}{\small{}0.132} & \textcolor{black}{\small{}-0.295} & \textcolor{black}{\small{}2.473} & \textcolor{black}{\small{}2.687} & \textcolor{black}{\small{}2.143} & \textcolor{black}{\small{}1.746}\tabularnewline
\textbf{\textcolor{black}{\small{}RMSE}} & \textcolor{black}{\small{}4.189} & \textcolor{black}{\small{}5.017} & \textcolor{black}{\small{}4.783} & \textcolor{black}{\small{}4.901} & \textcolor{black}{\small{}4.841} & \textcolor{black}{\small{}2.870} & \textcolor{black}{\small{}3.259} & \textcolor{black}{\small{}3.278} & \textcolor{black}{\small{}4.298} & \textcolor{black}{\small{}3.903} & \textcolor{black}{\small{}3.681} & \textcolor{black}{\small{}5.068} & \textcolor{black}{\small{}4.284} & \textcolor{black}{\small{}5.764} & \textcolor{black}{\small{}5.663}\tabularnewline
\textbf{\textcolor{black}{\small{}DM stat.}} &  & \textcolor{black}{\small{}3.168} & \textcolor{black}{\small{}2.621} & \textcolor{black}{\small{}2.931} & \textcolor{black}{\small{}3.112} &  & \textcolor{black}{\small{}3.976} & \textcolor{black}{\small{}3.579} & \textcolor{black}{\small{}2.659} & \textcolor{black}{\small{}1.887} &  & \textcolor{black}{\small{}4.285} & \textcolor{black}{\small{}1.377} & \textcolor{black}{\small{}3.180} & \textcolor{black}{\small{}2.898}\tabularnewline
\textbf{\textcolor{black}{\small{}$\triangle\%$}} &  & \textcolor{black}{\small{}16.5\%} & \textcolor{black}{\small{}12.4\%} & \textcolor{black}{\small{}14.5\%} & \textcolor{black}{\small{}13.5\%} &  & \textcolor{black}{\small{}11.9\%} & \textcolor{black}{\small{}12.4\%} & \textcolor{black}{\small{}33.2\%} & \textcolor{black}{\small{}26.5\%} &  & \textcolor{black}{\small{}27.4\%} & \textcolor{black}{\small{}14.1\%} & \textcolor{black}{\small{}36.1\%} & \textcolor{black}{\small{}35.0\%}\tabularnewline
\textbf{\textcolor{black}{\small{}MAE}} & \textcolor{black}{\small{}2.883} & \textcolor{black}{\small{}3.588} & \textcolor{black}{\small{}3.378} & \textcolor{black}{\small{}3.542} & \textcolor{black}{\small{}3.519} & \textcolor{black}{\small{}2.113} & \textcolor{black}{\small{}2.402} & \textcolor{black}{\small{}2.387} & \textcolor{black}{\small{}2.842} & \textcolor{black}{\small{}2.615} & \textcolor{black}{\small{}2.283} & \textcolor{black}{\small{}3.777} & \textcolor{black}{\small{}3.288} & \textcolor{black}{\small{}3.928} & \textcolor{black}{\small{}3.743}\tabularnewline
\textbf{\textcolor{black}{\small{}DM stat.}} &  & \textcolor{black}{\small{}3.127} & \textcolor{black}{\small{}2.370} & \textcolor{black}{\small{}3.095} & \textcolor{black}{\small{}3.951} &  & \textcolor{black}{\small{}3.995} & \textcolor{black}{\small{}3.333} & \textcolor{black}{\small{}4.222} & \textcolor{black}{\small{}3.137} &  & \textcolor{black}{\small{}6.273} & \textcolor{black}{\small{}4.484} & \textcolor{black}{\small{}7.322} & \textcolor{black}{\small{}6.394}\tabularnewline
\textbf{\textcolor{black}{\small{}$\triangle\%$}} &  & \textcolor{black}{\small{}19.7\%} & \textcolor{black}{\small{}14.6\%} & \textcolor{black}{\small{}18.6\%} & \textcolor{black}{\small{}18.1\%} &  & \textcolor{black}{\small{}12.0\%} & \textcolor{black}{\small{}11.5\%} & \textcolor{black}{\small{}25.6\%} & \textcolor{black}{\small{}19.2\%} &  & \textcolor{black}{\small{}39.6\%} & \textcolor{black}{\small{}30.6\%} & \textcolor{black}{\small{}41.9\%} & \textcolor{black}{\small{}39.0\%}\tabularnewline
 &  &  &  &  &  &  &  &  &  &  &  &  &  &  & \tabularnewline
\multicolumn{16}{l}{\textbf{Financial Crisis Period: 2007.01.03 - 2012.12.30}}\tabularnewline
\textbf{\small{}Model} & \textbf{\textcolor{black}{\small{}$\mathrm{RG}$}} & \textbf{\textcolor{black}{\small{}$\mathrm{EG}$}} & \textbf{\textcolor{black}{\small{}$\mathrm{G}$}} & \textbf{\textcolor{black}{\small{}$\mathrm{HN}$}} & \textbf{\textcolor{black}{\small{}$\mathrm{HN}_{\mathtt{vd}}$}} & \textbf{\textcolor{black}{\small{}$\mathrm{RG}$}} & \textbf{\textcolor{black}{\small{}$\mathrm{EG}$}} & \textbf{\textcolor{black}{\small{}$\mathrm{G}$}} & \textbf{\textcolor{black}{\small{}$\mathrm{HN}$}} & \textbf{\textcolor{black}{\small{}$\mathrm{HN}_{\mathtt{vd}}$}} & \textbf{\textcolor{black}{\small{}$\mathrm{RG}$}} & \textbf{\textcolor{black}{\small{}$\mathrm{EG}$}} & \textbf{\textcolor{black}{\small{}$\mathrm{G}$}} & \textbf{\textcolor{black}{\small{}$\mathrm{HN}$}} & \textbf{\textcolor{black}{\small{}$\mathrm{HN}_{\mathtt{vd}}$}}\tabularnewline
\textbf{\small{}Bias} & \textcolor{black}{\small{}0.257} & \textcolor{black}{\small{}-3.043} & \textcolor{black}{\small{}-3.172} & \textcolor{black}{\small{}-3.048} & \textcolor{black}{\small{}-2.109} & \textcolor{black}{\small{}-0.294} & \textcolor{black}{\small{}-0.824} & \textcolor{black}{\small{}-0.484} & \textcolor{black}{\small{}-1.427} & \textcolor{black}{\small{}-0.321} & \textcolor{black}{\small{}-0.551} & \textcolor{black}{\small{}2.219} & \textcolor{black}{\small{}2.688} & \textcolor{black}{\small{}1.621} & \textcolor{black}{\small{}1.788}\tabularnewline
\textbf{\small{}RMSE} & \textcolor{black}{\small{}4.981} & \textcolor{black}{\small{}5.867} & \textcolor{black}{\small{}5.744} & \textcolor{black}{\small{}5.649} & \textcolor{black}{\small{}5.779} & \textcolor{black}{\small{}3.383} & \textcolor{black}{\small{}3.972} & \textcolor{black}{\small{}3.944} & \textcolor{black}{\small{}5.497} & \textcolor{black}{\small{}4.866} & \textcolor{black}{\small{}4.668} & \textcolor{black}{\small{}6.087} & \textcolor{black}{\small{}4.985} & \textcolor{black}{\small{}7.258} & \textcolor{black}{\small{}7.318}\tabularnewline
\textbf{\small{}DM stat.} &  & \textcolor{black}{\small{}2.229} & \textcolor{black}{\small{}2.212} & \textcolor{black}{\small{}1.824} & \textcolor{black}{\small{}2.476} &  & \textcolor{black}{\small{}4.028} & \textcolor{black}{\small{}3.277} & \textcolor{black}{\small{}2.515} & \textcolor{black}{\small{}1.676} &  & \textcolor{black}{\small{}2.841} & \textcolor{black}{\small{}0.448} & \textcolor{black}{\small{}2.562} & \textcolor{black}{\small{}2.557}\tabularnewline
\textbf{\small{}$\triangle\%$} &  & \textcolor{black}{\small{}15.1\%} & \textcolor{black}{\small{}13.3\%} & \textcolor{black}{\small{}11.8\%} & \textcolor{black}{\small{}13.8\%} &  & \textcolor{black}{\small{}14.8\%} & \textcolor{black}{\small{}14.2\%} & \textcolor{black}{\small{}38.5\%} & \textcolor{black}{\small{}30.5\%} &  & \textcolor{black}{\small{}23.3\%} & \textcolor{black}{\small{}6.4\%} & \textcolor{black}{\small{}35.7\%} & \textcolor{black}{\small{}36.2\%}\tabularnewline
\textbf{\small{}MAE} & \textcolor{black}{\small{}3.340} & \textcolor{black}{\small{}4.341} & \textcolor{black}{\small{}4.220} & \textcolor{black}{\small{}4.210} & \textcolor{black}{\small{}4.383} & \textcolor{black}{\small{}2.557} & \textcolor{black}{\small{}3.069} & \textcolor{black}{\small{}2.946} & \textcolor{black}{\small{}3.798} & \textcolor{black}{\small{}3.327} & \textcolor{black}{\small{}2.849} & \textcolor{black}{\small{}4.552} & \textcolor{black}{\small{}3.868} & \textcolor{black}{\small{}4.908} & \textcolor{black}{\small{}4.954}\tabularnewline
\textbf{\small{}DM stat.} &  & \textcolor{black}{\small{}2.978} & \textcolor{black}{\small{}2.761} & \textcolor{black}{\small{}2.819} & \textcolor{black}{\small{}4.246} &  & \textcolor{black}{\small{}4.254} & \textcolor{black}{\small{}2.820} & \textcolor{black}{\small{}3.906} & \textcolor{black}{\small{}2.541} &  & \textcolor{black}{\small{}4.309} & \textcolor{black}{\small{}2.683} & \textcolor{black}{\small{}5.542} & \textcolor{black}{\small{}5.344}\tabularnewline
\textbf{\small{}$\triangle\%$} &  & \textcolor{black}{\small{}23.1\%} & \textcolor{black}{\small{}20.8\%} & \textcolor{black}{\small{}20.6\%} & \textcolor{black}{\small{}23.8\%} &  & \textcolor{black}{\small{}16.7\%} & \textcolor{black}{\small{}13.2\%} & \textcolor{black}{\small{}32.7\%} & \textcolor{black}{\small{}23.1\%} &  & \textcolor{black}{\small{}37.4\%} & \textcolor{black}{\small{}26.3\%} & \textcolor{black}{\small{}41.9\%} & \textcolor{black}{\small{}42.5\%}\tabularnewline
 &  &  &  &  &  &  &  &  &  &  &  &  &  &  & \tabularnewline
\multicolumn{16}{l}{\textbf{\small{}Post-Crisis Period: 2013.01.03-2018.12.30}}\tabularnewline
\textbf{\textcolor{black}{\small{}Model}} & \textbf{\textcolor{black}{\small{}$\mathrm{RG}$}} & \textbf{\textcolor{black}{\small{}$\mathrm{EG}$}} & \textbf{\textcolor{black}{\small{}$\mathrm{G}$}} & \textbf{\textcolor{black}{\small{}$\mathrm{HN}$}} & \textbf{\textcolor{black}{\small{}$\mathrm{HN}_{\mathtt{vd}}$}} & \textbf{\textcolor{black}{\small{}$\mathrm{RG}$}} & \textbf{\textcolor{black}{\small{}$\mathrm{EG}$}} & \textbf{\textcolor{black}{\small{}$\mathrm{G}$}} & \textbf{\textcolor{black}{\small{}$\mathrm{HN}$}} & \textbf{\textcolor{black}{\small{}$\mathrm{HN}_{\mathtt{vd}}$}} & \textbf{\textcolor{black}{\small{}$\mathrm{RG}$}} & \textbf{\textcolor{black}{\small{}$\mathrm{EG}$}} & \textbf{\textcolor{black}{\small{}$\mathrm{G}$}} & \textbf{\textcolor{black}{\small{}$\mathrm{HN}$}} & \textbf{\textcolor{black}{\small{}$\mathrm{HN}_{\mathtt{vd}}$}}\tabularnewline
\textbf{\textcolor{black}{\small{}Bias}} & \textcolor{black}{\small{}0.645} & \textcolor{black}{\small{}-2.326} & \textcolor{black}{\small{}-2.075} & \textcolor{black}{\small{}-2.304} & \textcolor{black}{\small{}-1.120} & \textcolor{black}{\small{}0.606} & \textcolor{black}{\small{}0.403} & \textcolor{black}{\small{}0.611} & \textcolor{black}{\small{}0.363} & \textcolor{black}{\small{}0.585} & \textcolor{black}{\small{}-0.039} & \textcolor{black}{\small{}2.729} & \textcolor{black}{\small{}2.686} & \textcolor{black}{\small{}2.667} & \textcolor{black}{\small{}1.705}\tabularnewline
\textbf{\textcolor{black}{\small{}RMSE}} & \textcolor{black}{\small{}3.206} & \textcolor{black}{\small{}3.989} & \textcolor{black}{\small{}3.571} & \textcolor{black}{\small{}4.014} & \textcolor{black}{\small{}3.670} & \textcolor{black}{\small{}2.242} & \textcolor{black}{\small{}2.339} & \textcolor{black}{\small{}2.434} & \textcolor{black}{\small{}2.591} & \textcolor{black}{\small{}2.602} & \textcolor{black}{\small{}2.302} & \textcolor{black}{\small{}3.782} & \textcolor{black}{\small{}3.442} & \textcolor{black}{\small{}3.709} & \textcolor{black}{\small{}3.250}\tabularnewline
\textbf{\textcolor{black}{\small{}DM stat.}} &  & \textcolor{black}{\small{}2.696} & \textcolor{black}{\small{}1.567} & \textcolor{black}{\small{}2.819} & \textcolor{black}{\small{}2.281} &  & \textcolor{black}{\small{}1.260} & \textcolor{black}{\small{}1.704} & \textcolor{black}{\small{}2.882} & \textcolor{black}{\small{}2.848} &  & \textcolor{black}{\small{}5.427} & \textcolor{black}{\small{}4.927} & \textcolor{black}{\small{}5.640} & \textcolor{black}{\small{}4.597}\tabularnewline
\textbf{\textcolor{black}{\small{}$\triangle\%$}} &  & \textcolor{black}{\small{}19.6\%} & \textcolor{black}{\small{}10.2\%} & \textcolor{black}{\small{}20.1\%} & \textcolor{black}{\small{}12.6\%} &  & \textcolor{black}{\small{}4.2\%} & \textcolor{black}{\small{}7.9\%} & \textcolor{black}{\small{}13.5\%} & \textcolor{black}{\small{}13.9\%} &  & \textcolor{black}{\small{}39.1\%} & \textcolor{black}{\small{}33.1\%} & \textcolor{black}{\small{}37.9\%} & \textcolor{black}{\small{}29.2\%}\tabularnewline
\textbf{\textcolor{black}{\small{}MAE}} & \textcolor{black}{\small{}2.426} & \textcolor{black}{\small{}2.835} & \textcolor{black}{\small{}2.536} & \textcolor{black}{\small{}2.874} & \textcolor{black}{\small{}2.653} & \textcolor{black}{\small{}1.669} & \textcolor{black}{\small{}1.735} & \textcolor{black}{\small{}1.828} & \textcolor{black}{\small{}1.885} & \textcolor{black}{\small{}1.902} & \textcolor{black}{\small{}1.715} & \textcolor{black}{\small{}3.003} & \textcolor{black}{\small{}2.709} & \textcolor{black}{\small{}2.948} & \textcolor{black}{\small{}2.532}\tabularnewline
\textbf{\textcolor{black}{\small{}DM stat.}} &  & \textcolor{black}{\small{}1.389} & \textcolor{black}{\small{}0.426} & \textcolor{black}{\small{}1.543} & \textcolor{black}{\small{}1.206} &  & \textcolor{black}{\small{}0.990} & \textcolor{black}{\small{}1.819} & \textcolor{black}{\small{}2.683} & \textcolor{black}{\small{}2.751} &  & \textcolor{black}{\small{}4.903} & \textcolor{black}{\small{}4.168} & \textcolor{black}{\small{}5.177} & \textcolor{black}{\small{}4.337}\tabularnewline
\textbf{\textcolor{black}{\small{}$\triangle\%$}} &  & \textcolor{black}{\small{}14.4\%} & \textcolor{black}{\small{}4.3\%} & \textcolor{black}{\small{}15.6\%} & \textcolor{black}{\small{}8.6\%} &  & \textcolor{black}{\small{}3.8\%} & \textcolor{black}{\small{}8.7\%} & \textcolor{black}{\small{}11.4\%} & \textcolor{black}{\small{}12.2\%} &  & \textcolor{black}{\small{}42.9\%} & \textcolor{black}{\small{}36.7\%} & \textcolor{black}{\small{}41.8\%} & \textcolor{black}{\small{}32.3\%}\tabularnewline
\bottomrule
\end{tabular}
\par\end{centering}
Note: Let the ``model error'' be the difference between the model-implied
quantity and the market-based quantity. It sample average is denoted
Bias and the rows indicated with ``$\triangle\%$'' state who much
larger the RMSE or MAE is for each alternative models, measured relative
to $\mathrm{RG}$. Diebold and Mariano statistics (DM stat. ) are
computed for the relative MSE or MAE losses, where the standard errors
are calculated with the Parzen kernel with $H=42$ as bandwidth. The
market VRP in this table is defined by $\mathrm{VRP}_{t}^{\mathtt{market}}=\mathrm{VIX}_{t}-\sqrt{\mathrm{\widehat{\mathrm{RVcc}}}_{t+1:t+22}}\times100$,
where the annualized volatility, $\sqrt{\mathrm{\widehat{\mathrm{RVcc}}}_{t+1:t+22}}\times100$,
is the predicted value from the HAR model.
\end{sidewaystable}

\subsection{Robustness Check: Estimation with logarithmic VIX\label{Sec:RobustlogVIX}}

The VIX index is occasionally very volatile and it reached very high
values during the financial crises. We inspect if some of our results
are driven by outliers. Rather than measuring pricing errors using
the level of VIX, we define VIX pricing errors for the logarithmically
transformed VIX. This transformation is a well known method for dimming
the influence of outliers. Table \ref{tab:parameters-logVIX} provides
parameters estimated with logVIX error specifications while in-sample
pricing performances are summarized in \ref{tab:In-sample-fit-logVIX}.
Again, we do not find significant changes in those results.
\begin{table}
\begin{centering}
{\small{}\caption{Full Sample Parameter Estimation with Alternative Pricing Errors for
log-VIX{\label{tab:parameters-logVIX}}}
}{\small\par}
\par\end{centering}
{\small{}\medskip{}
}{\small\par}
\begin{centering}
{\small{}{}}%
\begin{tabular}{c>{\centering}m{2.2cm}>{\centering}m{2.2cm}>{\centering}m{2.2cm}>{\centering}m{2.2cm}>{\centering}m{2.2cm}}
\toprule 
\textcolor{black}{\small{}Model} & {\small{}$\mathrm{RG}$} & {\small{}$\mathrm{EG}$} & {\small{}$\mathrm{G}$} & {\small{}$\mathrm{HN}$} & {\small{}$\mathrm{HN}_{\mathtt{vd}}$}\tabularnewline
\midrule 
\textcolor{black}{\small{}$\lambda$} & \textcolor{black}{\small{}0.014} & \textcolor{black}{\small{}0.137} & \textcolor{black}{\small{}0.303} & \textcolor{black}{\small{}4.673} & \textcolor{black}{\small{}4.969}\tabularnewline
 & \textcolor{black}{\small{}(0.036)} & \textcolor{black}{\small{}(0.010)} & \textcolor{black}{\small{}(0.021)} & \textcolor{black}{\small{}(0.339)} & \textcolor{black}{\small{}(0.796)}\tabularnewline
\textcolor{black}{\small{}$\omega$} & \textcolor{black}{\small{}-0.148} & \textcolor{black}{\small{}-0.120} & \textcolor{black}{\small{}1.60E-06} & \textcolor{black}{\small{}-6.08E-07} & \textcolor{black}{\small{}-9.63E-07}\tabularnewline
 & \textcolor{black}{\small{}(0.014)} & \textcolor{black}{\small{}(0.009)} & \textcolor{black}{\small{}(8.88E-08)} & \textcolor{black}{\small{}(1.54E-08)} & \textcolor{black}{\small{}(2.73E-08)}\tabularnewline
\textcolor{black}{\small{}$\beta$} & \textcolor{black}{\small{}0.985} & \textcolor{black}{\small{}0.986} & \textcolor{black}{\small{}0.930} & \textcolor{black}{\small{}0.955} & \textcolor{black}{\small{}0.950}\tabularnewline
 & \textcolor{black}{\small{}(0.001)} & \textcolor{black}{\small{}9.09E-04} & \textcolor{black}{\small{}(0.005)} & \textcolor{black}{\small{}(0.001)} & \textcolor{black}{\small{}(0.000)}\tabularnewline
\textcolor{black}{\small{}$\alpha$} &  &  & \textcolor{black}{\small{}0.064} & \textcolor{black}{\small{}1.70E-06} & \textcolor{black}{\small{}1.33E-06}\tabularnewline
 &  &  & \textcolor{black}{\small{}(0.004)} & \textcolor{black}{\small{}(1.84E-08)} & \textcolor{black}{\small{}(2.87E-09)}\tabularnewline
\textcolor{black}{\small{}$\delta$} &  &  &  & \textcolor{black}{\small{}150.120} & \textcolor{black}{\small{}182.711}\tabularnewline
 &  &  &  & \textcolor{black}{\small{}(2.206)} & \textcolor{black}{\small{}(0.988)}\tabularnewline
\textcolor{black}{\small{}$\tau_{1}$} & \textcolor{black}{\small{}-0.082} & \textcolor{black}{\small{}-0.052} &  &  & \tabularnewline
 & \textcolor{black}{\small{}(0.005)} & \textcolor{black}{\small{}(0.004)} &  &  & \tabularnewline
\textcolor{black}{\small{}$\tau_{2}$} & \textcolor{black}{\small{}0.018} & \textcolor{black}{\small{}0.103} &  &  & \tabularnewline
 & \textcolor{black}{\small{}(0.003)} & \textcolor{black}{\small{}(0.003)} &  &  & \tabularnewline
\textcolor{black}{\small{}$\gamma$} & \textcolor{black}{\small{}0.117} &  &  &  & \tabularnewline
 & \textcolor{black}{\small{}(0.010)} &  &  &  & \tabularnewline
\textcolor{black}{\small{}$\kappa$} & \textcolor{black}{\small{}0.303} &  &  &  & \tabularnewline
 & \textcolor{black}{\small{}(0.230)} &  &  &  & \tabularnewline
\textcolor{black}{\small{}$\phi$} & \textcolor{black}{\small{}1.063} &  &  &  & \tabularnewline
 & \textcolor{black}{\small{}(0.025)} &  &  &  & \tabularnewline
\textcolor{black}{\small{}$\delta_{1}$} & \textcolor{black}{\small{}-0.088} &  &  &  & \tabularnewline
 & \textcolor{black}{\small{}(0.012)} &  &  &  & \tabularnewline
\textcolor{black}{\small{}$\delta_{2}$} & \textcolor{black}{\small{}0.127} &  &  &  & \tabularnewline
 & \textcolor{black}{\small{}(0.010)} &  &  &  & \tabularnewline
\textcolor{black}{\small{}$\sigma^{2}$} & \textcolor{black}{\small{}0.308} &  &  &  & \tabularnewline
 & \textcolor{black}{\small{}(0.009)} &  &  &  & \tabularnewline
\textcolor{black}{\small{}$-\xi$} & \textcolor{black}{0.790} &  &  &  & \textcolor{black}{1.191}\tabularnewline
 & \textcolor{black}{\small{}(0.086)} &  &  &  & \textcolor{black}{\small{}(0.022)}\tabularnewline
\textcolor{black}{\small{}$\pi^{\mathbb{P}}$} & \textcolor{black}{\small{}0.985} & \textcolor{black}{\small{}0.986} & \textcolor{black}{\small{}0.994} & \textcolor{black}{\small{}0.994} & \textcolor{black}{\small{}0.995}\tabularnewline
{\small{}$\ell_{r}$} & \textcolor{black}{\small{}12895.03} & \textcolor{black}{\small{}12660.02} & \textcolor{black}{\small{}12522.64} & \textcolor{black}{\small{}12517.70} & \textcolor{black}{\small{}12651.56}\tabularnewline
{\small{}$\ell_{x}$} & \textcolor{black}{\small{}-3129.50} &  &  &  & \tabularnewline
\textcolor{black}{\small{}$\ell_{\mathtt{vix}}$} & \textcolor{black}{\small{}2360.77} & \textcolor{black}{\small{}1703.01} & \textcolor{black}{\small{}1883.67} & \textcolor{black}{\small{}1435.09} & \textcolor{black}{\small{}1574.06}\tabularnewline
\textcolor{black}{\small{}$\ell_{r,x}$} & \textcolor{black}{\small{}9765.52} &  &  &  & \tabularnewline
\textcolor{black}{\small{}$\ell_{r,\mathtt{vix}}$} & \textcolor{black}{\small{}15255.796} & \textcolor{black}{\small{}14363.024} & \textcolor{black}{\small{}14406.315} & \textcolor{black}{\small{}13952.796} & \textcolor{black}{\small{}14225.623}\tabularnewline
\textcolor{black}{\small{}$\ell_{r,x,\mathtt{vix}}$} & \textcolor{black}{\small{}12126.292} &  &  &  & \tabularnewline
\bottomrule
\end{tabular}
\par\end{centering}
\textcolor{black}{\small{}Note: Robust standard errors are in parenthesis.
}{\small{}The persistence parameter }\textcolor{black}{\small{}$\pi^{\mathbb{P}}$}{\small{}
is measured by $\beta+\phi\gamma$ for $\mathrm{RG}_{\mathtt{r}}$
and $\beta$ for all other models. }\textcolor{black}{\small{}For
the model }{\small{}$\mathrm{HN}_{\mathtt{vd}}$}\textcolor{black}{\small{}
we report the value of $(1+2\alpha\xi)^{-1}$ in place of $-\xi$
(the implied value for $\xi$ is here -}\textcolor{black}{121098.46}\textcolor{black}{\small{}).}{\small\par}
\end{table}
\begin{table}
\caption{In-sample Statistics of Model Performance with Alternative Error Specification\label{tab:In-sample-fit-logVIX}}

\medskip{}

\begin{centering}
\begin{tabular}{cccccc}
\toprule 
\textbf{\small{}Model} & \textbf{\small{}$\mathrm{RG}$} & \textbf{\small{}$\mathrm{EG}$} & \textbf{\small{}$\mathrm{G}$} & \textbf{\small{}$\mathrm{HN}$} & \textbf{\small{}$\mathrm{HN}_{\mathtt{vd}}$}\tabularnewline
\midrule 
\multicolumn{3}{l}{\textbf{\small{}Volatility Risk Premium}} &  &  & \tabularnewline
\midrule 
\textcolor{black}{\small{}Bias} & \textcolor{black}{\small{}-0.561} & \textcolor{black}{\small{}-4.092} & \textcolor{black}{\small{}-4.326} & \textcolor{black}{\small{}-4.627} & \textcolor{black}{\small{}-3.697}\tabularnewline
\textcolor{black}{\small{}RMSE} & \textcolor{black}{\small{}3.664} & \textcolor{black}{\small{}5.268} & \textcolor{black}{\small{}5.456} & \textcolor{black}{\small{}5.679} & \textcolor{black}{\small{}4.971}\tabularnewline
\textcolor{black}{\small{}$\triangle\%$} &  & \textcolor{black}{\small{}30.5\%} & \textcolor{black}{\small{}32.9\%} & \textcolor{black}{\small{}35.5\%} & \textcolor{black}{\small{}26.3\%}\tabularnewline
\textcolor{black}{\small{}MAE} & \textcolor{black}{\small{}2.547} & \textcolor{black}{\small{}4.509} & \textcolor{black}{\small{}4.712} & \textcolor{black}{\small{}4.944} & \textcolor{black}{\small{}4.105}\tabularnewline
\textcolor{black}{\small{}$\triangle\%$} &  & \textcolor{black}{\small{}43.5\%} & \textcolor{black}{\small{}45.9\%} & \textcolor{black}{\small{}48.5\%} & \textcolor{black}{\small{}37.9\%}\tabularnewline
 &  &  &  &  & \tabularnewline
\multicolumn{3}{l}{\textbf{\textcolor{black}{\small{}Volatility Index (VIX)}}} &  &  & \tabularnewline
\midrule 
\textcolor{black}{\small{}Bias} & \textcolor{black}{\small{}-0.202} & \textcolor{black}{\small{}-0.434} & \textcolor{black}{\small{}-0.398} & \textcolor{black}{\small{}-0.856} & \textcolor{black}{\small{}0.052}\tabularnewline
\textcolor{black}{\small{}RMSE} & \textcolor{black}{\small{}2.729} & \textcolor{black}{\small{}3.134} & \textcolor{black}{\small{}3.077} & \textcolor{black}{\small{}4.334} & \textcolor{black}{\small{}3.754}\tabularnewline
\textcolor{black}{\small{}$\triangle\%$} &  & \textcolor{black}{\small{}12.9\%} & \textcolor{black}{\small{}11.3\%} & \textcolor{black}{\small{}37.0\%} & \textcolor{black}{\small{}27.3\%}\tabularnewline
\textcolor{black}{\small{}MAE} & \textcolor{black}{\small{}1.958} & \textcolor{black}{\small{}2.287} & \textcolor{black}{\small{}2.191} & \textcolor{black}{\small{}2.528} & \textcolor{black}{\small{}2.386}\tabularnewline
\textcolor{black}{\small{}$\triangle\%$} &  & \textcolor{black}{\small{}14.4\%} & \textcolor{black}{\small{}10.6\%} & \textcolor{black}{\small{}22.5\%} & \textcolor{black}{\small{}17.9\%}\tabularnewline
 &  &  &  &  & \tabularnewline
\multicolumn{3}{l}{\textbf{\textcolor{black}{\small{}Annualized Volatility}}} &  &  & \tabularnewline
\midrule 
\textcolor{black}{\small{}Bias} & \textcolor{black}{\small{}0.359} & \textcolor{black}{\small{}3.658} & \textcolor{black}{\small{}3.928} & \textcolor{black}{\small{}3.771} & \textcolor{black}{\small{}3.748}\tabularnewline
\textcolor{black}{\small{}RMSE} & \textcolor{black}{\small{}3.531} & \textcolor{black}{\small{}4.608} & \textcolor{black}{\small{}4.524} & \textcolor{black}{\small{}5.984} & \textcolor{black}{\small{}5.665}\tabularnewline
\textcolor{black}{\small{}$\triangle\%$} &  & \textcolor{black}{\small{}23.4\%} & \textcolor{black}{\small{}22.0\%} & \textcolor{black}{\small{}41.0\%} & \textcolor{black}{\small{}37.7\%}\tabularnewline
\textcolor{black}{\small{}MAE} & \textcolor{black}{\small{}2.340} & \textcolor{black}{\small{}4.139} & \textcolor{black}{\small{}4.101} & \textcolor{black}{\small{}5.009} & \textcolor{black}{\small{}4.601}\tabularnewline
\textcolor{black}{\small{}$\triangle\%$} &  & \textcolor{black}{\small{}43.5\%} & \textcolor{black}{\small{}42.9\%} & \textcolor{black}{\small{}53.3\%} & \textcolor{black}{\small{}49.1\%}\tabularnewline
\bottomrule
\end{tabular}
\par\end{centering}
Note: Bias is defined as model generated value minus their market
counterpart; The rows starting with $\triangle\%$ present the increase
in percentage of RMSE and MAE for the competing models relative to
$\mathrm{RG}$.
\end{table}

\subsection{Robustness Check: Alternative Realized Measures\label{Sec:Robustalterrms}}

A wide range of realized measures of volatility have been proposed
since the use of the realized variance was popularized by \citet{AB1998}.
As alternative measures we use the realized variance based on both
5-minute and 10-minute returns, denoted $\mathrm{RV}{}_{5\mathrm{m}}$
and $\mathrm{RV}{}_{10\mathrm{m}}$, respectively. We also consider
the bi-power variation, $\mathrm{BV}{}_{5\mathrm{m}}$, by \citet{BN2004},
which is computed from 5-minute returns, and three different realized
kernels, denoted $\mathrm{RK}_{\mathrm{P}}$, $\mathrm{RK}_{\mathrm{TH}}$,
and $\mathrm{RK}_{\mathrm{B}}$, that are based on different kernel
function. $\mathrm{RK}_{\mathrm{P}}$ is based on the Parzen kerne,
$\mathrm{RK}_{\mathrm{TH}}$, on the modified Tukey-Hanning, which
is denoted $\mathrm{TH_{2}}$ in \citet{BHLS2008}, and $\mathrm{RK}_{\mathrm{B}}$
is the realized kernel estimator based on the Bartlett kernel function.
Finally, the median-based realized variance., by \citet{AndersenDobrevSchaumburg2012}
is denoted $\mathrm{MedRV}$. The realized measures were obtained
from the Realized Library at Oxford-Man institute (version 0.3). 

Table \ref{tab:In-sample-fit-AlterRMs} provides in-sample pricing
performances of the Realized GARCH model based on different realized
measures. The results are quite similar across with other realized
measures. For the VRP we observed some minor improvements over $\mathrm{RV}{}_{5\mathrm{m}}$,
but the relative differences are quite small. Interestingly, the two
jump robust estimators, $\mathrm{BV}{}_{5\mathrm{m}}$ and $\mathrm{MedRV}$,
are somewhat wors at fitting the VIX and Annualize Volatility, which
indicates that including the jump component is useful in this context.
However, for the VRP both jump-robust measures perform similarly to
other realized measures. The $\mathrm{RV}{}_{5\mathrm{m}}$ seems
to be adequate for the present modeling problem, as previously argued
in \citet{LiuPattonSheppard2015} for a range of problems. 
\begin{table}[H]
\caption{In-sample Statistics of Model Performance with Alternative Realized
Measures\label{tab:In-sample-fit-AlterRMs}}

\medskip{}

\begin{centering}
\begin{tabular}{cccccccc}
\toprule 
\textbf{Model} & $\mathrm{RV}{}_{5\mathrm{m}}$ & $\mathrm{RV}{}_{10\mathrm{m}}$ & $\mathrm{BV}{}_{5\mathrm{m}}$ & $\mathrm{RK}_{\mathrm{P}}$ & $\mathrm{RK}_{\mathrm{TH}}$ & $\mathrm{RK}_{\mathrm{B}}$ & $\mathrm{MedRV}$\tabularnewline
\midrule 
\multicolumn{3}{l}{\textbf{Volatility Risk Premium}} &  &  &  &  & \tabularnewline
\midrule 
Bias & -0.498 & -0.497 & -0.442 & -0.514 & -0.486 & -0.482 & -0.520\tabularnewline
RMSE & 3.825 & 3.838 & 3.818 & 3.843 & 3.764 & 3.764 & 3.844\tabularnewline
$\triangle\%$ &  & 0.34\% & -0.18\% & 0.48\% & -1.62\% & -1.60\% & 0.49\%\tabularnewline
MAE & 2.635 & 2.646 & 2.613 & 2.649 & 2.586 & 2.585 & 2.657\tabularnewline
$\triangle\%$ &  & 0.42\% & -0.81\% & 0.55\% & -1.90\% & -1.91\% & 0.84\%\tabularnewline
 & \multicolumn{1}{c}{} &  &  &  &  &  & \tabularnewline
\multicolumn{3}{l}{\textbf{Volatility Index (VIX)}} &  &  &  &  & \tabularnewline
\midrule 
Bias & -0.048 & -0.047 & -0.036 & -0.032 & -0.047 & -0.046 & -0.037\tabularnewline
RMSE & 2.504 & 2.520 & 2.668 & 2.575 & 2.548 & 2.547 & 2.808\tabularnewline
$\triangle\%$ &  & 0.63\% & 6.15\% & 2.77\% & 1.75\% & 1.72\% & 10.84\%\tabularnewline
MAE & 1.866 & 1.896 & 1.985 & 1.920 & 1.917 & 1.912 & 2.078\tabularnewline
$\triangle\%$ &  & 1.59\% & 6.00\% & 2.80\% & 2.64\% & 2.38\% & 10.20\%\tabularnewline
 & \multicolumn{1}{c}{} &  &  &  &  &  & \tabularnewline
\multicolumn{3}{l}{\textbf{Annualized Volatility}} &  &  &  &  & \tabularnewline
\midrule 
Bias & 0.449 & 0.450 & 0.406 & 0.481 & 0.439 & 0.436 & 0.483\tabularnewline
RMSE & 2.942 & 2.911 & 3.066 & 2.909 & 3.171 & 3.165 & 3.048\tabularnewline
$\triangle\%$ &  & -1.08\% & 4.02\% & -1.16\% & 7.22\% & 7.02\% & 3.47\%\tabularnewline
MAE & 2.002 & 1.977 & 2.130 & 2.013 & 2.190 & 2.185 & 2.134\tabularnewline
$\triangle\%$ &  & -1.28\% & 6.02\% & 0.53\% & 8.57\% & 8.40\% & 6.17\%\tabularnewline
\bottomrule
\end{tabular}
\par\end{centering}
Note: Empirical results based on alternative realized volatility measures.
$\mathrm{RV}{}_{5\mathrm{m}}$ and $\mathrm{RV}{}_{10\mathrm{m}}$
are the realized variances based on 5 minute returns and 10 minute
returns, respectively, $\mathrm{BV}{}_{5\mathrm{m}}$ is the bipower
variation based on five-minute returns,$\mathrm{RK}_{\mathrm{P}}$,
$\mathrm{RK}_{\mathrm{TH}}$, and $\mathrm{RK}_{\mathrm{B}}$ are
realized kernels estimator based the the Parzen, a modified Tukey-Hanning,
and the Bartlett kernel functions. Finally, $\mathrm{MedRV}$ is the
median based realized variance. The rows starting with $\triangle\%$
present the increase in percentage of RMSE and MAE for the competing
models relative to the one using 5min RV.
\end{table}

\newpage
\end{document}